\documentclass[a4paper]{article}

\usepackage{amsmath}
\usepackage{amsthm}
\usepackage{booktabs, makecell}

\usepackage{kpfonts}
\usepackage[T1]{fontenc}
\usepackage[hidelinks]{hyperref}
\usepackage[nocompress]{cite}

\usepackage{balance} % for balancing columns on the final page
\usepackage{nicefrac}
\theoremstyle{plain}

\usepackage{natbib}
\setcitestyle{numbers,square}

\usepackage{amssymb} 
\usepackage{enumitem}
\usepackage{nicefrac}
\usepackage{tikz}
\usepackage{comment}
\usepackage[ruled,vlined,linesnumbered]{algorithm2e}

\usepackage{thmtools} 
\usepackage{thm-restate}
\usepackage{soul}
\usepackage{multirow}

\usepackage[]{todonotes}

\newtheorem{theorem}{Theorem}
\newtheorem{proposition}[theorem]{Proposition}
\newtheorem{lemma}[theorem]{Lemma}
\newtheorem{corollary}[theorem]{Corollary}

\newtheorem{myex}{Example} 
\newenvironment{example}{\begin{myex}\rm}{\hfill$\vartriangle$\end{myex}}

\newtheorem{definition}{Definition}
\renewcommand{\leq}{\leqslant}
\renewcommand{\geq}{\geqslant}
\newcommand{\N}{N} %agents
\newcommand{\I}{I} %items aka. courses
\newcommand{\sat}{u} %utility
\newcommand{\prof}{\boldsymbol{u}} %profile
\newcommand{\reppi}[1]{\bar{\pi}^{(#1)}} %sequence of allocations
\newcommand{\oreppi}[1]{\bar{\rho}^{(#1)}} % different sequence of allocations
\newcommand{\areppi}[1]{\bar{\eta}^{(#1)}} % yet another different sequence of allocations

\newcommand{\bfx}{\boldsymbol{x}}
\newcommand{\bfy}{\boldsymbol{y}}

\title{Repeated Fair Allocation of Indivisible Items}
\date{\small\textsuperscript{1}The University of Tokyo,\quad
\textsuperscript{2}DBAI, TU Wien\\
\textsuperscript{3}CES, University of Paris 1 Panthéon-Sorbonne}

\author{
  Ayumi Igarashi\textsuperscript{1}\\{\small\href{mailto:igarashi@mist.i.u-tokyo.ac.jp}{\texttt{igarashi@mist.i.u-tokyo.ac.jp}}}
  \and
  Martin Lackner\textsuperscript{2}\\{\small\href{mailto:lackner@dbai.tuwien.ac.at}{\texttt{lackner@dbai.tuwien.ac.at}}}
  \and
  Oliviero Nardi\textsuperscript{2}\\{\small\href{mailto:oliviero.nardi@tuwien.ac.at}{\texttt{oliviero.nardi@tuwien.ac.at}}}
  \and
  Arianna Novaro\textsuperscript{3}\\{\small\href{mailto:arianna.novaro@univ-paris1.fr}{\texttt{arianna.novaro@univ-paris1.fr}}}
}
         
\newcommand{\BibTeX}{\rm B\kern-.05em{\sc i\kern-.025em b}\kern-.08em\TeX}

\begin{document}
\sloppy  % for right margin stuff

\maketitle

\begin{abstract}

The problem of fairly allocating a set of indivisible items is a well-known challenge in the field of (computational) social choice. In this scenario, there is a fundamental incompatibility between notions of fairness (such as envy-freeness and proportionality) and economic efficiency (such as Pareto-optimality). However, in the real world, items are not always allocated once and for all, but often repeatedly. For example, the items may be recurring chores to distribute in a household. Motivated by this, we initiate the study of the repeated fair division of indivisible goods and chores,
and propose a formal model for this scenario. In this paper, we show that, if the number of repetitions is a multiple of the number of agents, there always exists 
a sequence of allocations that is proportional and Pareto-optimal. On the other hand, irrespective of the number of repetitions, an envy-free and Pareto-optimal sequence of allocations may not exist.
For the case of two agents, we show that if the number of repetitions is even, it is always possible to find a sequence of allocations that is overall envy-free and Pareto-optimal. We then prove even stronger fairness guarantees, showing that every allocation in such a sequence satisfies some relaxation of envy-freeness.
Finally, in case that the number of repetitions can be chosen freely, we show that envy-free and Pareto-optimal allocations are achievable for any number of agents.
\end{abstract}

\section{Introduction}
%%%%%%%%%%%%%%%%%%%%%%%%%%%%%%%%%%%%%%%%%%%%%%%%%%%%%%%%%%%%%%%%%%%%%%%%%%%%%%%%%%%%%%%%%%%%%%%%%

In a variety of real-life scenarios, a group of agents has to divide a set of items among themselves.
These items can be desirable (goods) or undesirable (chores) and agents have (heterogeneous) preferences concerning them. In the case of \emph{goods}, we can think, for instance, of employees having to share access to some common infrastructure, like computing facilities.
In the case of \emph{chores}, we can think of roommates having to split household duties or teams having to split admin tasks. We may even have a set of \emph{mixed items}, where the agents may consider some items good and others bad: for instance, when assigning teaching responsibilities, some courses may be desirable to teach for some while undesirable (negative) for others.

The examples above are all instances of problems of fair allocation of indivisible items (see the recent survey by \citet{fairsurvey2022}). 
One of the challenges of fair division problems is that often, given the agents' preferences, it is impossible to find an allocation of the items which is both fair (e.g., no agent envies another agent's bundle) and efficient (e.g., no other allocation would make some agents better off, without making anyone worse off). However, a crucial feature which has so far received little attention, is that many fair allocation problems occur \emph{repeatedly}, in the sense that the same items will need to be assigned multiple times to the agents. For instance, university courses are usually offered every year, computing facilities may be needed daily or weekly by the same teams of employees, and household chores need to be done regularly.

In this paper, we thus focus precisely on those settings where a set of items has to be repeatedly allocated to a set of agents.
This opens the field to new and exciting research directions, revolving around the following central question: 
\begin{quote}
\emph{Can some fairness and efficiency notions %, which may be impossible to achieve in the standard setting of fair division, 
be guaranteed when taking a global perspective on the overall sequence of repeated allocations? Can we additionally achieve fairness and efficiency guarantees at each individual repetition? }
\end{quote}

%\noindent Additionally, we will check whether it is possible to guarantee fairness and efficiency globally (i.e., over all repetitions) while also guaranteeing them at each individual repetition.

%``Can some fairness and efficiency notions, which may be impossible to achieve in the standard setting of fair division,
% when focusing on one instance occurring at a specific point in time, 
%be guaranteed when taking a global perspective on the overall sequence of repeated allocations?''

\paragraph{Contribution and outline.} Our main contribution is the definition of a new (offline) model for the repeated fair allocation of goods and chores, which we present in Section~\ref{sec:model}. %\noteon{Should we highlight that we operate ``offline'' and not online? I.e., the allocation sequence is decided once and for all before the repetition starts.}
%\notean{Agree! Added a word here and a sentence in the rel. work.}
In particular, we specify how sequences of allocations will be evaluated with respect to classical axioms of fairness and efficiency: we distinguish between sequences satisfying some axioms \emph{per-round} (i.e., for every allocation composing them), or \emph{overall} (i.e., when considering the collection of all the bundles every agent has received in the sequence).
We consider two general cases: one where the number of repetitions is predetermined and given as part of the input and one where the number of repetitions can be chosen freely based on the instance. In the sequel, we refer to these as the \emph{fixed} and \emph{variable} cases, respectively.

%repeated allocations with a \emph{fixed}, predetermined number of rounds, and with a \emph{variable} (but still fixed) number of rounds, which can be chosen freely depending on the input.

%\notean{I find that ``variable (but still fixed)'' a bit confusing--- although we know what it means. What about ``repeated allocations where the number of rounds is given as part of the input, and those where it can be chosen freely.''? Should be then changed accordingly in the text (maybe not everywhere).}

In Section~\ref{sec:general}, we study our model for a fixed number of rounds ($k$) with $n$ agents. We show that: if the number of rounds is a multiple of $n$, we can always guarantee the existence of a sequence of allocations that is envy-free overall in (Proposition~\ref{prop:EF-with-cn-rounds}), and a sequence of allocations that is proportional and Pareto-optimal (Theorem~\ref{thm:generalCasePROPandPO}). For any number $n>2$ of agents and any fixed number $k$ of rounds,
This is essentially optimal in two respects. First, for any number $n \geq 2$ of agents and any fixed number $k$ of rounds which is not a multiple of $n$, a sequence of allocations that is proportional overall may not exist~(Proposition~\ref{prop:noMultipleNoPROP}). 
Moreover, for any number $n>2$ of agents and any fixed number $k$ of rounds, there is an instance of chore allocation where an envy-free and Pareto-optimal sequence of allocations may not exist (Theorem~\ref{thm:no-EF-PO-with-n-greater-than-2}). 
%Mention this is essentially optimal. 

In Section~\ref{sec:two}, we will stay in the fixed-$k$ setting but focus our model on the case of two agents. This scenario is fundamental in fair division, with numerous applications, including inheritance division, house-chore division, and divorce settlements~\cite{BramsFi00,brams2014two,kilgour2018two,housechore2023}. Here, we show that if the number of rounds is even, we can always find a sequence of allocations which is envy-free and Pareto-optimal overall, as well as per-round weak envy-free up to one item (Corollary~\ref{cor:POEFperroundwEF1}). Moreover, for two rounds, we can strengthen the per-round fairness guarantee to envy-freeness up to one item (EF1) (Corollary~\ref{cor:FPT-2-agents:EFPO}). At the cost of sacrificing the efficiency requirement, we also show that we can always find a sequence of allocations which is envy-free and per-round EF1 in polynomial time (Theorem~\ref{thm:2-agents-EF+per-round EF1}). These results turn out to be the best we can hope for. In Proposition~\ref{prop:no-ef1-perround-po}, we show that there is an instance with two agents and $k>2$ rounds where no sequence of allocations is envy-free and Pareto-optimal overall, as well as per-round EF1.

In Section~\ref{sec:variable}, we investigate the model with a variable number of rounds. Within this model, we can in fact achieve overall and per-round fairness guarantees for the case of $n$ agents. Specifically, we establish the existence of a sequence of allocations that satisfies envy-freeness and Pareto-optimality overall, while also meeting the per-round PROP[1,1] criterion (a weakening of proportionality). For scenarios involving goods only or chores only, we can achieve even per-round PROP1. To show this, we establish the connection between our model and the divisible and probabilistic fair division problems. 
An overview of our results can be found in Table~\ref{tab:results}.

\begin{table*}[t]
  \centering
    \begin{tabular}{p{1.5cm} l|ll}
        
        \multicolumn{2}{c|}{condition} & \multicolumn{2}{c}{result}\\
        agents ($n$) & rounds ($k$) & fairness guarantee & reference \\
        \toprule
        $n\geq 2$ & $k\in n\mathbb{N}$ & EF  overall  & \phantom{$\star$}Prop.~\ref{prop:EF-with-cn-rounds}  \\
        $n\geq 2$ & $k\notin n\mathbb{N}$ & \st{PROP overall} & \phantom{$\star$}Prop.~\ref{prop:noMultipleNoPROP}  \\
        $n> 2$ &$k\in \mathbb{N}$ & \st{EF+PO overall} & \phantom{$\star$}Thm.~\ref{thm:no-EF-PO-with-n-greater-than-2}\\
        $n\geq 2$ & $k\in n\mathbb{N}$ & PROP+PO overall & $\star$Thm.~\ref{thm:generalCasePROPandPO} \\
        $n\geq 2$ & variable~$k$ & \makecell[tr]{EF+PO overall, per-round PROP[1,1]\\(PROP1 for only goods/chores)} & $\star$Thm.~\ref{thm:variable:PO+EF+PROP2}\\
        \midrule
        % &$n=2$ & even & EF+PO overall & Thm.~\ref{thm:2andEvenEFandPO}\\
        $n=2$ & $k>2$ & \st{EF+PO overall, per-round EF1} & \phantom{$\star$}Prop.~\ref{prop:no-ef1-perround-po}\\
        $n=2$ & $k=2$ & EF+PO overall, per-round EF1 & $\star$Cor.~\ref{cor:FPT-2-agents:EFPO}\\
        $n=2$ & even & EF+PO overall, per-round weak EF1 & $\star$Cor.~\ref{cor:POEFperroundwEF1}\\
        $n=2$ & even & EF overall, per-round EF1 & \phantom{$\star$}Thm.~\ref{thm:2-agents-EF+per-round EF1}\\
        \bottomrule
    \end{tabular}
    
    \caption{Overview of our results regarding envy-freeness (EF), envy-freeness up to one item (EF1), weak EF1, proportionality (PROP), and Pareto-optimality (PO). Crossed-out results cannot be guaranteed under the stated conditions. We write $k\in n\mathbb{N}$ to denote that the number of rounds is a (fixed) multiple of $n$. The main positive results are highlighted with $\star$.}
    \label{tab:results}
\end{table*}

\paragraph{Related work.}

\citet{aziz2022fair} analyze fairness concepts for the allocation of indivisible goods and chores. Based on some of these results, \citet{housechore2023} have also developed an app to help couples divide household chores fairly. Another relevant application is that of the fair allocation of papers to reviewers in peer-reviewed conferences \cite{meir2021market,PayanZick22}, as well as the allocation of students to courses under capacity constraints \cite{othman2010finding}.
These articles focus however on ``one-shot'' (non-repeated) problems. 

Repetitions have been studied in the context of matchings \cite{hosseini2015matching,gollapudi2020almost,caragiannis2022repeatedly}. This line of work
uses models similar to ours. However, results do not carry over due to differences between matchings and arbitrary multi-unit assignments (as we consider).
Repeated fair allocation is also related to probabilistic fair division \cite{budish2013designing,aziz2023best}. We make this connection precise in Section~\ref{sec:variable}.

Various settings fall under the umbrella of \emph{dynamic} or \emph{online} fair decision-making \cite{AleksandrovW20,kohler2014value,kash2014no,freeman2018dynamic,benade2018make,zeng2020fairness}. 
% Both \citet{kash2014no} and \citet{freeman2018dynamic} study fair division problems where resources from a common pool are dynamically redistributed to a set of agents: the former assumes that the \emph{agents} may join the process over time, while the latter assumes that their \emph{demands} may vary over time.  \citet{benade2018make} and \citet{zeng2020fairness} study a setting where items may arrive over time, and they aim at decreasing envy among agents as time goes by. 
\citet{guo2009competitive} and \citet{cavallo2008efficiency} focus on a repeated setting, where a \emph{single} item must be allocated in every round. We work in an \emph{offline} setting, where agents have static and heterogeneous preferences over multiple items instead of demands, and the sets of agents and items is fixed. \citet{balan2011long} also study repeated allocations of items, but they focus on the average of utilities received by the agents for an allocation sequence.

In the context of elections, a closely related framework is that of \emph{perpetual voting}, introduced by~\citet{lackner2020perpetual}, where the agents participate in repeated elections to select a winning candidate and classical fairness axioms (as well as new ones) are introduced to evaluate aggregation rules with respect to sequences of elections. A similar approach has also been taken to analyze repeated instances of participatory budgeting problems \cite{lackner2021fairness}. 
%\citet{lackner2020perpetual} has introduced the framework of perpetual voting, where agents participate in repeated elections with approval ballots, and studied the fairness axioms of simple proportionality, independence of uncontroversial decisions, and bounded dry spells, for different aggregation rules.\footnote{\citet{lackner2021fairness} studied similar concepts for participatory budgeting, a setting where agents have to decide which projects to fund via a common budget. Note, however, that projects change every year, and they can end up unfunded---while the courses are usually stable and they all need to be taught.} In our case, however, the input and output of the problem are distinct, since agents should express preferences over courses, but the output is an allocation. 
\citet{freeman2017fair} consider a setting where at each round one alternative is selected, and agents' preferences may vary over time. A related model of simultaneous decisions, closer to fair division, has been studied by \citet{conitzer2017fair}.
%\citet{freeman2017fair} analyse the Nash social welfare in a setting where agents' preferences (expressed as an integer valuation over each alternative) may vary over time, and at each round one alternative must be elected.

\iffalse
Finally, we can see our work as belonging to a recent trend of research in computational social choice where a set of profiles needs to each be assigned an outcome---see the recent overview paper by \citet{boehmer2021broadening}. 
\fi

%%%%%%%%%%%%%%%%%%%%%%%%%%%%%%%%%%%%%%%%%%%%%%%%%%%%%%%%%%%%%%%%%%%%%%%%%%%%%%%%%%%%%%%%%%%%%%%%%
\section{The Model}\label{sec:model}
%%%%%%%%%%%%%%%%%%%%%%%%%%%%%%%%%%%%%%%%%%%%%%%%%%%%%%%%%%%%%%%%%%%%%%%%%%%%%%%%%%%%%%%%%%%%%%%%%

In this section, we present the model used throughout the paper. Furthermore, we recall some familiar concepts from the theory of fair division, and adapt them to our scenario.

We denote by $\N$ a finite set of $n$ \emph{agents}, who have to be assigned a set of $m$ \emph{items} in the finite set $\I$. An allocation $\pi \subseteq \N \times \I$ consists of agent-item pairs $(i,o)$, indicating that agent~$i$ is assigned item $o$. We denote by $\pi_i = \{o\in\I \colon (i, o) \in \pi\}$ the set of items that an agent~$i$ receives in allocation $\pi$. We assume that the allocation must be \emph{exhaustive}: all items must be assigned to some agent (and no two agents may receive the same item). Thus, $\bigcup_{i \in \N}\pi_i = \I$ and, for all distinct $i,j\in\N$, $\pi_i\cap\pi_j=\emptyset$.
We write $[k]$ to denote $\{1, \dots, k\}$. Finally, given a positive integer $\ell$, we denote by $\ell \mathbb{N}$ the set of positive integers multiples of~$\ell$. 

%%%%%%%%% DISUTILITIES
\paragraph{Utilities.} Each agent $i\in\N$ is associated with a (dis)utility function $\sat_i : \I \to \mathbb{R}$, which indicates how much they like or dislike each item. Namely, we consider a setting where each agent may view each item as a good, a chore, or a null item. In particular, we say that an item $o\in\I$ is an \emph{objective good} (resp. \emph{chore} or \emph{null}) if, for all $i\in\N$, $\sat_i(o)>0$ (resp. $\sat_i(o)<0$ or $\sat_i(o)=0$). Otherwise, we say that $o$ is a \emph{subjective} item.

We focus on additive utility functions and with a slight abuse of notation we write $u_i(S)=\sum_{o\in S}u_i(o)$ for the utility that agent~$i$ gets from set $S \subseteq \I$. Thus, the utility of agent~$i$ for an allocation $\pi$ is given by $\sat_i(\pi_i)$.
We denote by $\prof = (\sat_1, \dots, \sat_n)$ the profile of utilities for the agents. 

%%%%%%%%%%%% FAIRNESS

\paragraph{Fairness.} We introduce several fairness concepts, as defined by \citet{aziz2022fair} and by \citet{fairsurvey2022}. 
%\noteml{shouldn't we cite something else here? maybe a handbook? these are all standard concepts after all}

A classical notion of fairness is that of \emph{envy-freeness}: an allocation is envy-free if no agent finds that the bundle given to someone else is better than the one they received. 

\begin{definition}[EF]
For agents $\N$, items $\I$, and profile $\prof$, allocation $\pi$ is \emph{envy-free (EF)} if for any $i,j\in \N$, $u_i(\pi_i){\geq}~u_i(\pi_j)$.
\end{definition}

It is easy to see that this notion is too strong and cannot always be achieved (consider the case of one objective good and two agents desiring it). Thus, it has then been relaxed to \emph{envy-freeness up to one item}, which has been further generalized to take into account both goods and chores.

\begin{definition}[EF1]
For agents $\N$, items $\I$, and profile $\prof$, an allocation $\pi$ is \emph{envy-free up to one item (EF1)} if for any $i,j\in \N$, either $\pi$ is envy-free, or there is $o \in \pi_i \cup \pi_j$ such that $u_i(\pi_i \setminus \{o\})\geq u_i(\pi_j\setminus \{o\})$.
\end{definition}

Yet another concept is that of \emph{proportionality} of an allocation, where each agent receives their due share of utility.

\begin{definition}[PROP] 
For agents $\N$, items $\I$, and profile $\prof$, an allocation $\pi$ satisfies \emph{proportionality} if for each agent $i\in\N$ we have $\sat_i(\pi_i) \geq \nicefrac{\sat_i(\I)}{n}$.
\end{definition}
Note that envy-freeness implies proportionality when assuming additive utilities---see, e.g., \citet[Prop. 1]{aziz2022fair}.

Finally, in a similar spirit to EF1, we can weaken proportionality to the notion of \emph{proportionality up to one item} (PROP1) as well as its relaxed version (PROP[1,1]), to guarantee that agents receive a share of utility close to their proportional fair share. The notion of PROP1 has been proposed by \citet{conitzer2017fair} and later extended to the mixed setting by \citet{aziz2022fair}. 

\begin{definition} 
An allocation $\pi$ is said to satisfy 
\begin{itemize}
\item \emph{PROP1} if for each $i\in\N$, $\sat_i((\pi_i \setminus X)\cup Y) \geq \nicefrac{\sat_i(\I)}{n}$ for some $X \subseteq \pi_i$ and $Y \subseteq I \setminus \pi_i$ with $|X \cup Y| \leq 1$.

\item \emph{PROP[1,1]} if for each $i\in\N$, $\sat_i((\pi_i \setminus X)\cup Y) \geq \nicefrac{\sat_i(\I)}{n}$ for some $X \subseteq \pi_i$ and $Y \subseteq I \setminus \pi_i$ with $|X|,|Y| \leq 1$.
\end{itemize}
\end{definition}

We note that PROP1 and PROP[1,1] coincide with each other when all items have non-negative values, since removing an item from an envious agent’s bundle does not increase his or her utility. The same relation holds when all items have non-positive values.  
Note that \citet{ShoshanSeHa23} recently introduced EF[1,1], which is an analogous version of our PROP[1,1] for envy-freeness. 

\paragraph{Efficiency.} Alongside fairness, it is often desirable to distribute the items as efficiently as possible. One way to capture efficiency is via the notion of \emph{Pareto-optimality}, meaning that no improvement to the current allocation can be made without hurting some agent.

\begin{definition}[PO]
For agents $\N$, items $\I$, and profile $\prof$, an allocation $\pi$ is \emph{Pareto-optimal (PO)} if there is no other allocation $\rho$ such that for all $i\in\N$, $\sat_i(\rho_i) \geq \sat_i(\pi_i)$ and for some $j\in\N$ it holds that $\sat_j(\rho_j) > \sat_j(\pi_j)$. If such a $\rho$ exists, we say that it \emph{Pareto-dominates} $\pi$.
\end{definition}

%%%%%%%%% REPEATED SETTING

\paragraph{Repeated setting.} In our paper, we will be interested in \emph{repeated} allocations of the items to the agents, i.e., sequences of allocations. We denote by $\reppi{k}=(\pi^1, \pi^2, \dots, \pi^k)$ the repeated allocation of the $m$ items in $\I$ to the $n$ agents in $\N$ over $k$ time periods (or rounds). More formally, we will consider $k$ copies of the set $\I$, such that $\I^1 = \{o^1_1, \dots, o^1_m\},$ $ \dots, \I^k = \{o^k_1, \dots, o^k_m\}$. Then, each $\pi^\ell$ corresponds to an allocation of the items in $\I^\ell$ to the agents in $\N$. For all agents $i \in \N$, all items $o \in \I$ and all $\ell\in[k]$, we let the utility be unchanged: i.e., $u_i(o) = u_i(o^\ell)$. When clear from context, we will drop the superscript $\ell$ from the items.

As a first approach, we will assess fairness over time by considering the global set of items that each agent has received across the $k$ rounds. We denote by $\pi^{\cup k}$ the allocation of $k\cdot m$ (copies of the) items to the $n$ agents, where $\pi^{\cup k}_i = \pi_i^1 \cup \dots \cup \pi_i^k$ for each $i \in \N$. Namely, we consider the allocation $\pi^{\cup k}$ where each agent gets the bundle of (the copies of) items that they have received across all the $k$ time periods in $\reppi{k}$. We say that a sequence of allocations $\reppi{k}$ for some $k$ satisfies an axiom \emph{overall}, if $\pi^{\cup k}$ satisfies it (when clear from context, we omit the term ``overall''). Similarly, we say that $\reppi{k}$ Pareto-dominates $\oreppi{k}$ overall if $\pi^{\cup k}$ Pareto-dominates $\rho^{\cup k}$, and that $\reppi{k}$ is Pareto-optimal overall if no $\oreppi{k}$ dominates it.

Since envy-freeness implies proportionality, we get:
\begin{proposition}
    For additive utilities, if $\reppi{k}$ is envy-free overall, then it is proportional overall.\label{prop:EF-impl-prop}
\end{proposition}

As a second approach, we will assess fairness over time by checking whether each repetition satisfies some desirable property. For an axiom $A$, we say that a sequence of allocations $\reppi{k}$ (for some $k$) satisfies \emph{per-round} $A$, if for every $j\in[k]$, the allocation $\pi^j$ satisfies the axiom. For example, an allocation $\reppi{k}=(\pi^1,\dots,\pi^k)$ is \emph{per-round} EF1 if every allocation $\pi^j$ for all $j \in [k]$ is EF1.

\begin{proposition}\label{prop:overall:perround:PO}
If $\reppi{k}$ is Pareto-optimal overall, then it is per-round Pareto-optimal.
\end{proposition}

The converse direction does not hold, namely, per-round Pareto-optimality may not imply overall Pareto-optimality, as the following example shows. 

\begin{example}\label{ex:repeat2}
Consider a utility profile over two agents and two items where $\sat_1(o_1) = 4$, $\sat_1(o_2) = 5$, $\sat_2(o_1) = 3$ and $\sat_2(o_2) = 9$. Now consider a four-round allocation where agent~$1$ gets both items in the first and second round, and agent~$2$ gets both items in the remaining two rounds. It is easy to verify that such an allocation sequence is per-round PO. However, it is not PO overall. Indeed, this gives utilities $18$ for agent~$1$ and $24$ for agent~$2$. Instead, an allocation sequence where agent~$1$ always gets $o_1$ and gets $o_2$ once (thus, agent~$2$ takes $o_2$ thrice and no item in one round) yields utilities $21$ to agent~$1$ and $27$ to agent~$2$. Observe that we could obtain a similar example where all items are chores by just multiplying all utilities by $-1$.
\end{example}

On the other hand, for envy-freeness and proportionality, we get the following.

\begin{proposition}\label{prop:overall:perround:EF}
For additive utilities, if $\reppi{k}$ is per-round envy-free (resp. proportional), then it is envy-free (resp. proportional) overall.
\end{proposition}

We can see that the converse does not necessarily hold. Indeed, consider a two-agent, two-round scenario with one objective good, where we give the good to one agent in the first round, and to the other agent in the second round. This is envy-free (resp. proportional) overall, but not per-round.

Finally, note that our \emph{overall} results would also hold in the framework of one-shot allocations with $k$ clones of each item. However, the latter model cannot capture the \emph{per-round} results. Thus, our model cannot be reduced to the standard fair division setting with cloned items.

%%%%%%%%%%%%%%%%%%%%%%%%%%%%%
\section{The case of $n$ agents}\label{sec:general}
%%%%%%%%%%%%%%%%%%%%%%%%%%%%%

In this section, we study the possibility of finding fair and efficient sequences of allocations for the general case of any number of agents.

We start by looking at envy-freeness. First of all, observe that, whenever $k$ is a multiple of $n$, then we can always guarantee an EF allocation.
    
\begin{proposition}\label{prop:EF-with-cn-rounds}
If $k\in n\mathbb{N}$, there exists a sequence $\reppi{k}$ which satisfies envy-freeness overall.
\end{proposition}

\begin{proof}
Start from an arbitrary initial allocation. Then, let each agent give their current bundle to the next agent in the next allocation (agent $n$ gives theirs to agent~$1$). As each agent~$i$ receives each bundle exactly $\nicefrac{k}{n}$ times, allocation $\pi^{\cup k}$ is envy-free. Thus,  $\reppi{k}$ is envy-free overall.
\end{proof}

Note that this is not always possible if $k$ is not a multiple of $n$, irrespective of $m$.

\begin{restatable}{proposition}{noMultipleNoPROP}
    For every $n\geq 2$, every $k\in\mathbb{N}\setminus n\mathbb{N}$ and every $m\in\mathbb{N}$, an allocation that is proportional overall is not guaranteed to exist, even if the items are all goods or all chores.\label{prop:noMultipleNoPROP}
\end{restatable}
\begin{proof}
    Consider the following profile of utilities, where every agent has utility $1$ for all items, except for a special item $o^*$, for which all agents have utility $u^\star=k(n-1)(m-1)+1$.
    Consider any allocation sequence $\reppi{k}$. Since $k$ is not divisible by $n$, it is impossible to give $o^*$ an equal amount of times to all agents (i.e., $\nicefrac{k}{n}$ times). Thus, there must be some agent~$i$ receiving $o^*$ less times than $\nicefrac{k}{n}$. Observe that $i$ can receive $o^\star$ at most $\lfloor\nicefrac{k}{n}\rfloor<\nicefrac{k}{n}$ times. Then, her utility can be at most
    \[
        u_i(\pi^{\cup k}_i) \leq \lfloor\nicefrac{k}{n}\rfloor u^\star + k(m-1).
    \]
    We show that this is less than the proportional fair share due to $i$, that is:
    \begin{align*}
        &\lfloor\nicefrac{k}{n}\rfloor u^\star + k(m-1) < \nicefrac{k}{n}\cdot u_i(\I)\\
        \iff  &\lfloor\nicefrac{k}{n}\rfloor u^\star + k(m-1) < \nicefrac{k}{n}\cdot(u^\star+m-1) \\
        \iff &(\nicefrac{k}{n}-\lfloor\nicefrac{k}{n}\rfloor) u^\star > (k-\nicefrac{k}{n})(m-1) \\
       \iff  &u^\star > \frac{\nicefrac{1}{n}\cdot k(n-1)(m-1)}{\nicefrac{k}{n}-\lfloor\nicefrac{k}{n}\rfloor}.\\
    \end{align*}
    However, since $\nicefrac{k}{n}-\lfloor\nicefrac{k}{n}\rfloor\geq\nicefrac{1}{n}$, the above holds because
    \begin{align*}
        u^\star = k(n-1)(m-1)+1 & > \\
        k(n-1)(m-1) &\geq \frac{\nicefrac{1}{n}\cdot k(n-1)(m-1)}{\nicefrac{k}{n}-\lfloor\nicefrac{k}{n}\rfloor}.
    \end{align*}
    Therefore, $i$ receives less than her proportional fair share.
    Observe, finally, that we could obtain a similar example where all items are chores by just multiplying all utilities by $-1$.
\end{proof}

%Intuitively, in a profile where every agent has utility $1$ for every item, except for a desirable item $o^*$ for which all agents have high utility $k(n-1)(m-1)+1$, there must be some agent~$i$ receiving $o^*$ less times than $\nicefrac{k}{n}$ (as $k$ is not divisible by $n$). 

If we also require Pareto-optimality, we get the following.

\begin{restatable}{theorem}{noEFPOwithNgreaterThanTwo}
    For all $n>2$ and $k\in\mathbb{N}$, there exist a set of items~$\I$ and a profile~$\prof$ such that no allocation is envy-free and Pareto-optimal overall. The set $I$ can be chosen to contain only two objective chores.\label{thm:no-EF-PO-with-n-greater-than-2}
\end{restatable}
\begin{proof}
    First, observe that, by Propositions~\ref{prop:EF-impl-prop} and \ref{prop:noMultipleNoPROP}, if $k$ is not a multiple of $n$, we are done. Now let $k$ be a multiple of $n$. Consider the following profile, where $s$ and $b$ stand for a \emph{small chore} and a \emph{big chore}, respectively:
    
 \begin{table}[h]
    \begin{center}
    \begin{tabular}{c|cc}
    & $s$ & $b$ \\
    \midrule
    $1$& $-1$&$-k$\\
    $2$& $-1$&$-k$\\
        \vdots & \vdots & \vdots \\
    ${n-1}$& $-1$&$-k$\\
    $n$& $-1$&$-\nicefrac{k}{n}$\\
    \end{tabular}
    \end{center}
    \end{table}

Let $p=\nicefrac{k}{n}$. We claim that in every sequence of allocations that is EF overall, all agents must receive $s$ and $b$ the same amount of times (namely $p$ times). Before proving this claim, observe that any such allocation always has a Pareto-improvement. Indeed, suppose that agent $1$ gives one of her allocations of $b$ to agent~$n$ and agent~$n$ gives  back to agent~$1$ all of her allocations of $s$. Then, agent~$1$ will be happier (she increased her utility by $k-\nicefrac{k}{n}$), and all other agents will be equally happy (for $i\in\{2,\dots,n-1\}$ nothing changes, while $n$ gains $\nicefrac{k}{n}$ but also loses $\nicefrac{k}{n}$ utility). However, agent~$2$, e.g., now envies agent~$1$.

    Let us now prove the claim that all agents receive both items $p$ times. Let $s_i$ and $b_i$ be the number of times where agent~$i$ receive items $s$ and $b$, respectively. We will first prove that, for every $i,j\in[n-1]$, $s_i=s_j$ and $b_i=b_j$ must hold. Consider agents $1$ and $2$, and suppose toward a contradiction that $b_1<b_2$. To guarantee envy-freeness, we must have $-s_1-kb_1=-s_2-kb_2$, which means that $s_1-s_2=k(b_2-b_1)$. 
    Since $s_1,s_2,b_1,b_2\in\{0,\dots,k\}$ and $b_1<b_2$, this holds if and only if $s_1=k$ and $s_2=0$ and $b_2=b_1+1$. Furthermore, note that $s_1=k$ implies that for every other agent $j\neq 1$, $s_j=0$. Thus, to guarantee envy-freeness, we must have that, for all agents $i\in\{3,\dots,n-1\}$, $b_i=b_2=b_1+1$ holds as well (as $s_i=0=s_2$).

    Now, for agent~$2$ not to envy agent $n$, we must have $-b_2 k\geq -b_n k$ and thus $b_2 \leq b_n$. For agent $n$ not to envy agent~$2$, we obtain similarly that $b_2\geq b_n$. Thus $b_1+1=b_2=\cdots=b_n$, which implies: 
\begin{align*}
    &\sum_{i\in[n]}b_i=k,\\
    \implies & b_1+(n-1)(b_1+1)=k,\\
   \implies   &nb_1 +(n-1) =k,\\
   \implies  &p=b_1+\frac{(n-1)}{n}. 
\end{align*}

Since $p,n,b_1\in\mathbb{N}$ and $n>2$, we finally have our contradiction: thus, $b_1\geq b_2$. By a completely symmetrical argument we get that $b_1=b_2$ must hold (which implies $s_1=s_2$ by envy-freeness). Furthermore, we can repeat this argument for every pair of agents $i$ and $j$ with $i,j\in[n-1]$. Thus, every such pair of agents must have $b_i=b_j$ and $s_i=s_j$.

    We are ready to prove that all agents must receive both items the same amount of times, namely $p$ times. First, for agent~$1$ not to envy agent $n$, we must have:
\begin{align*}
    &-s_1-kb_1\geq -s_n-kb_n,\\
    \implies &-s_1-kb_1\geq - (k-(n-1)s_1)-k(k-(n-1)b_1),\\
    \implies &s_1-\frac{k}{n} \leq k(\frac{k}{n}-b_1),\\
    \implies &s_1-p \leq np(p-b_1).
\end{align*}
    With a similar line of reasoning, we can show that, for agent $n$ not to envy agent~$1$, we must have that $s_1-p\geq p(p-b_1)$.   
Additionally, we know that $(n-1)s_1\leq k$ (as we cannot assign $s$ to the agents in $[n-1]$ more than $k$ times), which gives $s_1\leq p\frac{n}{n-1}$. We can now show the following (given $p\geq 1$):
\begin{align*}
    &s_1-p\geq p(p-b_1),\\
    \iff  &(p-b_1) \leq \frac{s_1}{p}-1\leq \frac{n}{n-1}-1=\frac{1}{n-1}\leq\frac{1}{2}. 
\end{align*}

Since $p-b_1\in\mathbb{Z}$, $p>b_1$ would imply $1\leq p-b_1\leq \frac{1}{2}$, a contradiction: hence, $p\leq b_1$. By the above, we also know that $p(p-b_1) \leq s_1-p\leq  np(p-b_1)$, and thus:
\begin{align*}
    &p(p-b_1) \leq np(p-b_1),\\
    \implies & (p-b_1)\leq n(p-b_1),
\end{align*}
which implies $p \geq b_1$. 
    Hence, $b_1=p=\nicefrac{k}{n}$. Finally:
    \begin{align*}
        p(p-b_1)=0\leq s_1-p \leq np(p-b_1)=0.
    \end{align*}
   This gives $s_1=b_1=p=\nicefrac{k}{n}$, which in turn yields $s_1=b_1=\dots=s_n=b_n=\nicefrac{k}{n}$. This concludes the proof.
\end{proof}

The following example shows that, even if $n=k=3$ and all items are good, there are utility profiles with no EF and PO allocations.

\begin{example}\label{ex:n_k_3_noEF_PO}

Consider the following profile for three agents $\N = \{1,2,3\}$ expressing their utilities over two goods $\I = \{o_1, o_2\}$:

    \begin{table}[h]
    \begin{center}
    \begin{tabular}{c|cc}
    & $o_1$ & $o_2$\\
    \midrule
    $1$& $1$&$2$\\
    $2$& $1$&$2$\\
    $3$& $1$&$1$\\
    \end{tabular}
    \end{center}
    \end{table}
    
    Assume $k=3$. One can verify that, in order to be achieve envy-freeness, all agents must receive both goods exactly once. This gives total utilities $3$, $3$ and $2$ for agents $1$, $2$ and $3$, respectively. However, any such allocation sequence is dominated by a sequence where agent $1$ gets $o_2$ twice, agent $2$ gets both items once, and agent $2$ receives $o_1$ twice. Indeed, this gives total utilities $4$, $3$ and $2$ for agents $1$, $2$ and $3$, respectively. However, here agent $2$ envies agent $1$. Thus, no envy-free and PO allocation exists for $k=3$.
\end{example}

In light of Theorem~\ref{thm:no-EF-PO-with-n-greater-than-2}, envy-freeness seems too strong of a requirement, even in the repeated setting. However, we can at least always find a proportional and Pareto-optimal sequence of allocations.

\begin{theorem}
    For every $n\geq2$ and $k\in n\mathbb{N}$, an allocation sequence that is proportional and Pareto-optimal overall always exists, and can be computed in time $O(n^{mk}\cdot m)$.\label{thm:generalCasePROPandPO}
\end{theorem}

\begin{proof}
    Consider any proportional sequence of allocations $\reppi{k}$ (which must exist, by Propositions~\ref{prop:EF-impl-prop} and~\ref{prop:EF-with-cn-rounds}). Clearly, any Pareto-improvement over $\reppi{k}$ must also be proportional. We get the following algorithm:

    \begin{enumerate}
    \item Initialize a variable $r$ as $\reppi{k}$, as defined previously.
    \item Iterate over all possible $k$-round allocations $\oreppi{k}$, and do: 
    \begin{itemize}
    \item  if $\oreppi{k}$ Pareto-dominates $r$, set $r$ to $\oreppi{k}$.
    \end{itemize}
    \item Return $r$.
    \end{enumerate}

    This algorithm runs in $O(n^{mk}\cdot m)$. Indeed, every iteration at Step $2$ of the algorithm requires time $O(m)$, as we just need to sum up the utilities for the items received by each agent to compare the total utilities of $r$ and $\oreppi{k}$. Furthermore, there are $n^{mk}$ possible allocations, as for any round $j\in [k]$ and every object $o\in\I$, we could assign object $o$ in round $j$ to $n$ different agents. It remains to be shown that the algorithm is correct. 

    Let us call $\areppi{k}$ the allocation in $r$ returned by the algorithm. As argued above, $\areppi{k}$ must be proportional overall, since it is either equal to $\reppi{k}$ or a Pareto-improvement of it. Thus, suppose toward a contradiction that $\areppi{k}$ is not PO. Then, there must be some Pareto-optimal $\oreppi{k}$ dominating $\areppi{k}$ that is encountered before $\areppi{k}$ in the iteration---since, if it is encountered after, the variable $r$ would be updated to $\oreppi{k}$. But since Pareto-dominance is transitive, we know that $\oreppi{k}$ Pareto-dominates all the allocations corresponding to the values that the variable $r$ took before being updated to $\areppi{k}$, and hence $r$ is updated to $\oreppi{k}$ during iteration. However, this means that $r$ cannot be updated to $\areppi{k}$, as $\oreppi{k}$ dominates it: a contradiction. This concludes the proof.
\end{proof}

Note that the existence guarantee in Theorem~\ref{thm:generalCasePROPandPO} is tight because of Proposition~\ref{prop:noMultipleNoPROP}. Moreover, we can define an integer linear program (ILP) to compute such a proportional and PO allocation (Figure~\ref{fig:ilp-n-agents}).

\begin{figure}[bht]
    \centering
    \noindent
    %\fbox{%
        \parbox{0.8 \columnwidth}{%
            \begin{align*}
                \text{\textbf{maximize} \ } &\sum_{i\in N}\sum_{o\in\I} u_i(o) \cdot x_o^i &\\
                \text{\textbf{subject to}: \ }& x_o^i \in \{0,\dots,k\} && \text{ for }o\in\I,i\in N &\\
                & \sum_{i\in N} x_o^i = k  && \text{ for }o\in\I &\\
                & \sum_{o\in\I} u_i(o)\cdot x_o^i\geq\nicefrac{k}{n}\cdot u_i(\I) &&  \text{ for }i\in N &
        \end{align*}}
        %}
    \caption{An ILP for finding a proportional and PO allocation sequence for $n$ agents.
        In this ILP, we maximize the social welfare and thus guarantee PO.
        Variable $x_o^i$ indicates how often agent~$i$ receives item~$o$.
    If $k\notin n\mathbb{N}$, this ILP may be unsatisfiable due to Proposition~\ref{prop:noMultipleNoPROP}.}\label{fig:ilp-n-agents}
\end{figure}

Based on this, we present some fixed-parameter tractability results. A problem is said to be fixed-parameter tractable (FPT) with respect to a parameter $\tau$ if each instance $H$ of this problem can be solved in time $f(\tau){poly}(|H|)$ for some function $f$.
%We start by looking at the problem of computing proportional and PO allocations (for $n$ agents). 
In practical applications, relevant parameters are the number of agents $n$, the number of items $m$, and the maximum value of an item $u_{max}=\max_{i \in N, o \in \I}u_{i}(o)$. Invoking a result of \citet{bredereckEC2019} for one-shot fair division, as well as our ILP program (Figure~\ref{fig:ilp-n-agents}) we can get fixed-parameter tractability with respect to (combinations of) these parameters.

\begin{theorem}
    For every $n\geq2$ and $k\in n\mathbb{N}$, finding a proportional and Pareto-optimal allocation sequence can be done in FPT time with respect to $n+m$, and when each utility $u_i(o)$ is an integer in FPT time with respect to $n+u_{max}$.\label{thm:n+umax+FPT-any-agent}
\end{theorem}

\begin{proof}
Recall that such a proportional and PO allocation sequence is guaranteed to exist (Theorem~\ref{thm:generalCasePROPandPO}). Fixed-parameter tractability in $n+m$ follows from the ILP in Figure~\ref{fig:ilp-n-agents}, together with the fact that solving an ILP is FPT in the number of variables ($nm$) \citep{lenstra1983integer}. Next, \citet{bredereckEC2019} showed that the problem of computing a PO \emph{maxmin allocation} $\pi$ that maximizes the minimum utility $\min_{i \in N}u_i(\pi_i)$ can be done in FPT time with respect to $n+u_{max}$ when each utility $u_i(o)$ is an integer. Thus, one can compute an allocation $\pi^{\cup k}$ that maximizes the minimum utility $\min_{i \in N}u_i(\pi^{\cup k}_i)$ in FPT time with respect to $n+u_{max}$.
\end{proof}

%%%%%%%%%%%%%%%%%%%%%%%%%%%%%
\section{The case of two agents}\label{sec:two}
%%%%%%%%%%%%%%%%%%%%%%%%%%%%%
In this section, we consider the case of repeated fair division among two agents. Several authors have focused on this special but important case (see, e.g., the papers by \citet{brams2012undercut}, \citet{brams2014two}, \citet{kilgour2018two}, \citet{aziz2022fair}, \citet{ShoshanSeHa23}, and the references therein). This captures some practically relevant problems, for example, house chore division among couples~\cite{housechore2023}, or divorce settlements~\cite{BramsTaylor1996}.

A first intuitive idea for the repeated setting could be to apply twice the \emph{Generalized Adjusted Winner} algorithm introduced by \citet{aziz2022fair}, as in the one-shot setting it returns efficiently a PO and EF1 allocation of the goods and chores, by alternating the roles of the winner and the loser.
However, this approach may fail. The following example shows that, when computing EF and PO allocations for two agents, the approached based on the Generalized Adjusted Winner (GAW) algorithm does not work.

\begin{example}\label{ex:adjusted}
Consider the following profile for two agents $\N = \{1,2\}$ expressing their utilities over three goods $\I = \{o_1, o_2, o_3\}$:

    \begin{table}[h]
    \begin{center}
    \begin{tabular}{c|ccc}
    & $o_1$ & $o_2$ & $o_3$\\
    \midrule
    $1$& $4.5$&$3$&$7$\\
    $2$& $9$&$5$&$10$
    \end{tabular}
    \end{center}
    \end{table}
    
    Consider now a two-round sequence $\reppi{2} = (\pi^1, \pi^2)$ where we apply the GAW algorithm by choosing agent~$1$ as the loser in the first round and agent~$2$ as the loser in the second round. We then get allocation $\pi^1$ with $\pi^1_1=\{o_3\}$ and $\pi^1_2=\{o_1,o_2\}$, and an allocation $\pi^2$ with $\pi^2_1=\{o_2,o_3\}$ and $\pi^2_2=\{o_1\}$. Both $\pi_1$ and $\pi_2$ are, by themselves, EF1 and PO (as we used GAW each time). However, $\reppi{2}$ is not EF, since agent~$2$ envies agent~$1$. By multiplying all the utilities by $-1$, we obtain an analogous example for the case where all items are chores.
    \end{example}

Despite this, we will show that, whenever the number of rounds is even, we can still always find an allocation that is PO and EF overall. However, we lose the guarantee of an efficient computation. First, we need the following fact:

\begin{restatable}{proposition}{twoEFPROPrelation}
    If $n=2$, for additive utilities, proportionality implies envy-freeness.\label{prop:2n-prop-impl-EF}
\end{restatable}
\begin{proof}
    Let $\N=\{1, 2\}$ and consider some proportional allocation $\reppi{k}$. Then we get:
    \begin{align*}
        &u_1(\pi^{\cup k}_1) \geq \nicefrac{k}{2}\cdot u_1(\I) \\
        \implies &k\cdot u_1(\I)-u_1(\pi^{\cup k}_2) \geq \nicefrac{k}{2}\cdot u_1(\I) \\
        \implies  &u_1(\pi^{\cup k}_2) \leq \nicefrac{k}{2}\cdot u_1(\I) \leq u_1(\pi^{\cup k}_1) 
    \end{align*}
    Thus, agent~$1$ does not envy agent~$2$. A symmetrical argument shows that $2$ does not envy agent~$1$. Hence $\reppi{k}$ is envy-free overall.
\end{proof}

Now, we can show the following.

\begin{theorem}\label{thm:EFPO}
    If $n=2$ and $k\in 2\mathbb{N}$, an allocation sequence that is envy-free and Pareto-optimal overall always exists, and can be computed in time $O(m\cdot (k+1)^m)$.\label{thm:2andEvenEFandPO}
\end{theorem}
\begin{proof}
 From the algorithm in the proof of Theorem~\ref{thm:generalCasePROPandPO} and from Proposition~\ref{prop:2n-prop-impl-EF} we get an envy-free and Pareto-optimal allocation. Since $n=2$, there are $(k+1)^m$ possible allocations (up to symmetry-breaking), as each agent receives each of the $m$ items either $0$, $1$, \dots, or $k$ times. Thus, we obtain a run time of $O(m\cdot (k+1)^m)$.
\end{proof}

For two agents, as an immediate consequence of Theorem~\ref{thm:n+umax+FPT-any-agent} and Proposition~\ref{prop:2n-prop-impl-EF}, we get the following.

\begin{proposition}\label{prop:FPT-2-agents}
    If $n=2$ and $k\in 2\mathbb{N}$, an allocation which is Pareto-optimal and envy-free can be computed in FPT time with respect to $m$, and when each utility $u_i(o)$ is an integer in FPT time with respect to $u_{max}$.
\end{proposition}

Theorem~\ref{thm:EFPO} provides strong overall fairness and efficiency guarantees, but it does not ensure per-round fairness. Indeed, an envy-free and PO allocation that is per-round EF1 may not exist, even if all items are goods or chores.

\begin{restatable}{proposition}{noEFonePerroundPO}
If $n=2$ and $k>2$, an allocation that is per-round EF1, Pareto-optimal overall, and envy-free is not guaranteed to exist, even when all items are goods or chores.\label{prop:no-ef1-perround-po}
\end{restatable}
\begin{proof}
    First, observe that the number of rounds $k$ needs to be even, otherwise we know that EF cannot always be achieved (Proposition~\ref{prop:noMultipleNoPROP}). Next, consider the following utility profile over two items $\I = \{o_1, o_2\}$:

    \begin{table}[h]  
    \begin{center} 
    \begin{tabular}{c|ccc}
    & $o_1$ & $o_2$\\
    \midrule
    $1$& $1$&$3$\\
    $2$& $1$&$2$
    \end{tabular}
     \end{center}
    \end{table}
      
    For $k>2$, there is no EF and PO allocation sequence which is also per-round EF1. Indeed, to be per-round EF1, both agents should receive at least one item in each round: otherwise, the agent receiving no items would envy the other for more than one item. Hence, in each round, each agent will receive either $o_1$ or $o_2$ (but not both). Then, in any sequence of allocations where one of the agents receives $o_1$ more than $\nicefrac{k}{2}$ times, they will envy the other agent. Hence, we can discard them and only focus on the case where agent $1$ (resp. agent $2$) receives both items $\nicefrac{k}{2}$ times.

    Call this sequence $\reppi{k}$. We show that $\reppi{k}$ is not PO. In fact, it is dominated by the sequence where agent~$1$ gets only $o_1$ exactly $\nicefrac{k}{2}-2$ times, only $o_2$ exactly $\nicefrac{k}{2}+1$ times, and no items once. Indeed, the former sequence $\reppi{k}$ gives satisfaction $2k$ and $\nicefrac{3}{2}k$ to agent~$1$ and agent~$2$ (respectively), while the latter gives satisfaction $2k+1$ and $\nicefrac{3}{2}k$. Hence, there is no EF and PO allocation that is per-round EF1.

    Finally, we can make a counter-example where all items are chores by multiplying all utilities by $-1$.
\end{proof}

% \ayumi{I think we can actually get EF + per-round EF1 for two agents and $2N$ rounds. So, Proposition~\ref{prop:no-ef1-perround-po} needs to consider PO. }

Nevertheless, for two agents and two rounds, i.e., $n=k=2$, it is always possible to achieve PO and EF as well as per-round EF1. Indeed, in this case, we can always transform an allocation that is PO and EF overall (which is guaranteed to exist for two agents) to an allocation that is per-round EF1 while preserving the PO and EF guarantees.  
A similar idea (i.e., exchanging items among two agents while preserving some initial property) was also used by \citet{ShoshanSeHa23}.

% \noteon{I added this sentence as well.}

% \noteon{Should I reformat this to have the same structure of Corollary \label{cor:POEFperroundwEF1}? (that is, one theorem about turning an EF and PO allocation into a EF+PO+per-round EF1 allocation, and one corollary with the full runtime))}

\begin{restatable}{theorem}{thmEFonePerroundPOtwo}\label{thm:POEFEF1:two}
    Suppose $k=n=2$. Given an allocation sequence that is Pareto-optimal and envy-free, an allocation sequence which is Pareto-optimal, envy-free, and per-round EF1 can be computed in polynomial time.
\end{restatable}
\begin{proof}
Consider an overall PO and EF allocation sequence $\reppi{2}$ for two agents. Let $\I_1=\pi^1_1\cap\pi^2_1$ (and similarly for $\I_2$) be the items assigned to agent $1$ (resp. $2$) in both rounds. Moreover, let $O=\I\setminus(\I_1\cup\I_2)$ be the items that each agent receives once. Observe that, by PO, we can assume that every subjective item $o$ is always assigned to the agent $i\in\{1,2\}$ for whom $u_i(o)>u_j(o)$ (where $j$ is the other agent). Furthermore, we can remove objective null items from consideration. Hence, all items in $O$ are objective goods or chores. Let $O^+\subseteq O$ be the objective goods in $O$, and let $O^-\subseteq O$ be the objective chores in $O$. Clearly, $O^+\cap O^-=\emptyset$.

We start by assigning all chores in $O$ to agent~$1$ in the first round (and to agent~$2$ in the second round), and conversely for the goods. Then, we refine this sequence as follows. We progressively move the chores from agent~$1$ to agent~$2$ in the first round, and vice versa in the second round, and conversely for the goods. We stop when the allocation becomes per-round EF1. The formal description of the algorithm is given as follows. 

\begin{algorithm}
\caption{Algorithm for computing a PO and EF allocation that is per-round EF1 }
\label{alg:POEF1:k=2}
\SetInd{0.8em}{0.3em}
{\bf Input}: A PO and EF allocation $\reppi{2}=({\pi}^1,{\pi}^2)$\;
Let $\I_1=\pi^1_1\cap\pi^2_1$ and $\I_2=\pi^1_2\cap\pi^2_2$\;
Let $O=I \setminus (I_1 \cup I_2)$\;
Let $O^+ \subseteq O$ be the set of objective goods in $O$ and $O^- \subseteq O$ the set of objective chores in $O$\;
Update $\reppi{2}=({\pi}^1,{\pi}^2)$ so that $\pi^1_1=\I_1\cup O^-$, $\pi^1_2=\I_2\cup O^+$, $\pi^2_1=\I_1\cup O^+$, and $\pi^2_2=\I_2\cup O^-$\;
Let $O=\{o_1,o_2,\ldots,o_t\}$ and initialize $j=1$\;
\While{$\reppi{2}$ is not per-round EF1}{
\If{$o_j$ is an objective chore}{
Remove $o_j$ from $\pi^1_1$ and $\pi^2_2$, and add it to $\pi^1_2$ and $\pi^2_1$\;
}
\Else{
Remove $o_j$ from $\pi^1_2$ and $\pi^2_1$, and add it to  $\pi^1_1$ and $\pi^2_2$\;
}
Update $j=j+1$\;
}
\Return $\reppi{2}$;
\end{algorithm}

    \begin{figure}[h!]
\centering
\begin{tikzpicture}
     \node (pi1) at (1.5,1.75) {$\pi^1$};
    \draw (0,0) rectangle (3,1.5); 
        \draw (1.5,0) -- (1.5,1.5);
     \draw[dashed] (0.1,0.1) rectangle (1.4,0.65); 
     \node (A11) at (0.75, 0.4) {$O^-$};
    \draw [->] (1.1,0.4) -- (2.25,0.4) node[midway,below] {$o$} ;
            \draw[dashed] (1.6,0.8) rectangle (2.9,1.4);      \node (A12) at (2.25, 1.1) {$O^+$};
    \draw [->] (1.8,1) -- (0.8,1) node[midway,above] {$o'$} ;
    
    \node (pi11) at (0.75,-0.25) {$\pi^1_1$};
    \node (pi12) at (2.25,-0.25) {$\pi^1_2$};
    
    \node (pi2) at (5.5,1.75) {$\pi^2$};
    \draw (4,0) rectangle (7,1.5); 
            \draw (5.5,0) -- (5.5,1.5);
         \draw[dashed] (4.1,0.8) rectangle (5.4,1.4);      \node (A21) at (4.75, 1.1) {$O^+$};
    \draw [->] (5.2,1) -- (6.25,1) node[midway,above] {$o'$} ;   
                \draw[dashed] (5.6,0.1) rectangle (6.9,0.65);      \node (A22) at (6.25, 0.4) {$O^-$};
    \draw [->] (5.8,0.4) -- (4.8,0.4) node[midway,below] {$o$} ;
     \node (pi21) at (4.75,-0.25) {$\pi^2_1$};
    \node (pi22) at (6.25,-0.25) {$\pi^2_2$};
\end{tikzpicture}
\caption{Exchanges of goods and chores between agents in the {\bf while}-loop of the algorithm.}
\end{figure}
    
Algorithm~\ref{alg:POEF1:k=2} runs in polynomial time and it remains to show that it is correct.

    First, observe that during the execution of the algorithm, $\reppi{2}$ remains an allocation that is PO and EF overall, since we never change the total amount of times an agent receives an item w.r.t. the initial allocation. Thus, if the initial allocation of Algorithm~\ref{alg:POEF1:k=2} is per-round EF1, we are done. We assume in the following that the initial allocation is not per-round EF1, which implies that at least one agent envies another in some round. Furthermore, observe that, by Proposition~\ref{prop:overall:perround:PO}, $\reppi{2}$ satisfies per-round PO as well. This fact will be useful later.
    
    By overall EF, we show the following fact, which will also be useful later:
    \begin{align*}
        2u_1(\I_1)+u_1(O) & \geq 2u_1(\I_2)+u_1(O) \\
        \implies u_1(\I_1) & \geq u_1(\I_2)
    \end{align*}
    Similarly, we can derive $u_2(\I_2)\geq u_2(\I_1)$. This implies that 
    $u_1(\I_1 \cup O^+)\geq u_1(\I_2 \cup O^-)$
    and 
    $u_2(\I_2 \cup O^+)\geq u_2(\I_1 \cup O^-)$ since all items in $O^+$ are objective goods and all items in $O^-$ are objective chores. 
    
    Since $u_1(\I_1 \cup O^+)\geq u_1(\I_2 \cup O^-)$, if the algorithm moves all the objective chores from agent $1$ to agent $2$ and all the objective goods from agent $2$ to agent $1$ in the first round, agent $1$ does not envy agent $2$ with respect to the first round. Moreover, her utility in the first round never decreases, so that once agent $1$ becomes envy-free in the first round, she remains envy-free in the latter steps. 
    Similarly, if the algorithm moves all the objective chores from agent $2$ to agent $1$ and all the objective goods from agent $1$ to agent $2$ in the second round, agent $2$ does not envy agent $1$ with respect to the second round. In addition, once agent $2$ becomes envy-free with respect to the second round, she remains envy-free in the latter steps. 

Thus, there exists an item $o_j$ (with $j\geq 1$) such that, after transferring $o_j$, agent $1$ does not envy agent $2$ in the first round, and agent $2$ does not envy agent $1$ in the second round. Let $o_j$ be the first such item. Recall that the initial allocation of Algorithm~\ref{alg:POEF1:k=2} is not per-round EF1. Thus, at least one agent $i$ envies the other agent in the $i$-th round before transferring $o_j$. We assume without loss of generality that before transferring $o_j$, agent~$1$ envies agent~$2$ in the first round. 
We further assume that $o_j$ is an objective chore (the case when $o_j$ is an objective good is analogous).

Now denote $\oreppi{2}=(\rho^1,\rho^2)$ by the allocation in the algorithm that appears before transferring $o_j$ and $\areppi{2}=(\eta^1,\eta^2)$ by the allocation after transferring $o_j$. Then, there must be two disjoint sets $L, R\subseteq O\setminus\{o_j\}$ with the following properties.
\begin{itemize}
        \item In the first round $\rho^1$, agent~$1$ gets $\I_1\cup L\cup\{o_j\}$ and agent~$2$ gets $\I_2\cup R$.
        \item In the second round $\rho^2$, agent~$1$ gets $\I_1\cup R$ and agent~$2$ gets $\I_2\cup L\cup\{o_j\}$.
        \item In $\areppi{2}$, $o_j$ is transferred from agent~$1$ to agent~$2$ in the first round, and $o_j$ is transferred from agent~$2$ to agent~$1$ in the second round (i.e., 
        $\eta^1_1=\I_1\cup L$, $\eta^1_2=\I_2\cup R \cup \{o_j\}$, $\eta^2_1=\I_1\cup R \cup \{o_j\}$, and $\eta^2_2=\I_2\cup L$). 
    \end{itemize} 

Observe that, by choice of $o_j$, agent~$1$ does not envy agent~$2$ in the first round of $\areppi{2}$, i.e., $u_1(\I_1\cup L)\geq u_1(\I_2\cup R\cup\{o_j\})$.
Further, agent~$2$ does not envy agent~$1$ in the second round of $\areppi{2}$, i.e., $u_2(\I_2\cup L)\geq u_2(\I_1\cup R\cup\{o_j\})$. 

We claim that at least one of the allocations, $\oreppi{2}$ and $\areppi{2}$, is per-round EF1 and thus the algorithm correctly finds a desired allocation. There are four cases:

    \begin{itemize}
        \item \textbf{Case 1}: $u_1(\I_1\cup L)\geq u_1(\I_2\cup R)$ and $u_2(\I_2\cup L)\geq u_2(\I_1\cup R)$. In this case, $\oreppi{2}$ is per-round EF1. Indeed, agent~$1$ is envious in the first round, so by EF (overall), she cannot be envious in the second round. By removing $o_j$ from her bundle in the first round, we eliminate envy by agent $1$, since $u_1(\I_1\cup L)\geq u_1(\I_2\cup R)$. Furthermore, agent~$2$ cannot be envious in the first round (by per-round PO, since agent~$1$ already is envious). Finally, if agent~$2$ is envious in the second round, we can eliminate her envy by removing $o_j$ from her bundle, as $u_2(\I_2\cup L)\geq u_2(\I_1\cup R)$.
        
        \item \textbf{Case 2}: $u_1(\I_1\cup L)< u_1(\I_2\cup R)$ and $u_2(\I_2\cup L)< u_2(\I_1\cup R)$. 
        We claim that $\areppi{2}$ is per-round EF1. By the choice of $o_j$, agent $1$ does not envy agent $2$ in the first round and agent $2$ does not envy agent $1$ in the second round. Moreover, since $u_1(\I_1\cup L)< u_1(\I_2\cup R)$ and $u_1(I_1) \geq u_1(I_2)$, $u_1(L) < u_1(R)$ and thus we have $u_1(\I_2\cup L)<u_1(\I_1\cup R)$. Similarly, we have $u_2(\I_1\cup L)< u_2(\I_2\cup R)$. 
        Thus, if we remove $o_j$ from the bundle of agent $1$ (resp. $2$) in the second round (resp. first) we eliminate envy, since $u_1(\I_2\cup L)<u_1(\I_1\cup R)$ and $u_2(\I_1\cup L)< u_2(\I_2\cup R)$.

        \item \textbf{Case 3}: $u_1(\I_1\cup L)< u_1(\I_2\cup R)$ and $u_2(\I_2\cup L)\geq u_2(\I_1\cup R)$. We show that $\areppi{2}$ is per-round EF1. Again, agent $1$ does not envy agent $2$ in the first round and agent $2$ does not envy agent $1$ in the second round. Further, since $u_1(\I_1\cup L)< u_1(\I_2\cup R)$ and $u_1(I_1) \geq u_1(I_2)$, $u_1(L) < u_1(R)$ and thus we have $u_1(\I_2\cup L)<u_1(\I_1\cup R)$. Hence, if we remove $o_j$ from the bundle of agent~$1$ in the second round, she does not envy agent~$2$. It remains to show that agent $2$ is not envious in the first round once we remove $o_j$ from her bundle. Toward a contradiction, assume that $u_2(\I_2\cup R)<u_2(\I_1\cup L)$. Now, suppose we change the first round as follows:
        \begin{itemize}
            \item Agent $1$ gets $\I_2\cup R$ instead of $\I_1\cup L$ and
            \item Agent $2$ gets $\I_1\cup L\cup\{o_j\}$ instead of $\I_2\cup R\cup\{o_j\}$.
        \end{itemize}
        Observe that, since $u_1(\I_1\cup L)< u_1(\I_2\cup R)$ and $u_2(\I_2\cup R)<u_2(\I_1\cup L)$, this is a Pareto-improvement of $\areppi{2}$, a contradiction. Thus, $u_2(\I_2\cup R)\geq u_2(\I_1\cup L)$, and agent~$2$ is not envious in the first round if we remove $o_j$ from her allocation.

        \item \textbf{Case 4}: $u_1(\I_1\cup L)\geq u_1(\I_2\cup R)$ and $u_2(\I_2\cup L)< u_2(\I_1\cup R)$. Similarly to the previous argument, we can show that $\areppi{2}$ is per-round EF1.

    \end{itemize}

    This concludes the proof.
    %In all cases, we reach a per-round EF1 allocation. Therefore, during the execution of the algorithm, we eventually obtain a per-round EF1 allocation, which we return. This concludes the proof.
\end{proof}

Combining the above with Theorem~\ref{thm:EFPO}, we obtain the following corollary. 

\begin{corollary}\label{cor:FPT-2-agents:EFPO}
If $k=n=2$, then an allocation which is Pareto-optimal, envy-free, and per-round EF1 always exists. Moreover, it can be computed in time $O(m\cdot3^m)$.
\end{corollary}

%\ayumi{Moved Example 2 to the model section.}

As we have seen in Proposition \ref{prop:no-ef1-perround-po}, we cannot guarantee per-round EF1 together with PO and EF overall for a more general number of rounds $k \in 2 \mathbb{N}$ with $k>2$. Nevertheless, these properties become compatible if we relax a per-round fairness requirement to the following (weaker) version of EF1, where the envy of $i$ toward $j$ can be eliminated by either giving a good of $j$ to $i$, or imposing a chore of $i$ on $j$. 
        
\begin{definition}[Weak EF1]
An allocation $\pi$ is \emph{weak EF1} if for all $i,j \in N$, either $u_i(\pi_i) \geq u_i(\pi_j)$ or there is an item $o \in \pi_i \cup \pi_j$ such that $u_i(\pi_i\cup \{o\}) \geq u_i(\pi_j\setminus \{o\})$ or $u_i(\pi_i\setminus \{o\}) \geq u_i(\pi_j\cup \{o\})$.
\end{definition}

Again, to prove the following theorem, we show that it is always possible to transform an allocation that is PO and EF overall to an allocation that is per-round weak EF1 while preserving the PO and EF guarantees.
%\ayumi{What is the complexity of finding a Pareto-improvement for two agents?}

\begin{restatable}{theorem}{TheoremPOEFperroundwEF}\label{thm:POEFperroundwEF1}
Suppose $n=2$. Given a $k$-round allocation that is Pareto-optimal and envy-free, an allocation which is Pareto-optimal, envy-free, and per-round weak EF1 can be computed in polynomial time.
\end{restatable}

\begin{corollary}\label{cor:POEFperroundwEF1}
If $n=2$ and $k\in 2\mathbb{N}$, then an allocation which is Pareto-optimal, envy-free, and per-round weak EF1 always exists, and can be computed in time $O(m\cdot(k+1)^m)$.
\end{corollary}

In order to obtain Corollary~\ref{cor:POEFperroundwEF1}, it suffices to prove Theorem~\ref{thm:POEFperroundwEF1}, since we can apply Theorem \ref{thm:EFPO} to get an allocation sequence ${\bar \pi}^{(k)}$ that is PO and envy-free overall when $k\in 2\mathbb{N}$. 

To do so, suppose that we are given a PO and EF $k$-round allocation. Looking into each individual round, there may be an allocation $\pi^{\ell}$ that is not fair. Similarly to the proof of Theorem~\ref{thm:POEFEF1:two}, we thus repeatedly transfer an item between \emph{envious} rounds and \emph{envy-free} rounds while preserving the property that the overall allocation $\reppi{k}$ is PO and EF over the course of the algorithm. We can show that this process terminates in polynomial time and eventually yields a per-round weak EF1 allocation. The formal description of our algorithm is presented in Algorithm~\ref{alg:POEF1wEF1}.

%We prove that given an allocation that is PO and EF overall, Algorithm~\ref{alg:POEF1wEF1} transforms it into a per-round weak EF1 allocation while keeping the PO and EF properties.
Specifically, our adjustment procedure is divided into two phases: one that makes $\reppi{k}$ weak EF1 for agent~$1$ (Lines \ref{line:firstwhile} -- \ref{line:firstwhileend}), and another that makes $\reppi{k}$ weak EF1 for agent~$2$ (Lines \ref{line:secondwhile} -- \ref{line:secondwhileend}). In the first phase, it identifies $\pi^j$ and $\pi^i$ where agent~$1$ envies the other at $\pi^j$ but she does not envy at $\pi^i$ (Line~\ref{line:first:envy} and Line~\ref{line:first:EF}). 
It then finds either a good $o$ in $\pi^i_1 \setminus \pi^j_1$ or a chore in $\pi^j_1 \setminus \pi^i_1$ (Line~\ref{line:first:if}) and transfers the item between $\pi^i_1$ and $\pi^j_1$ without changing the overall allocation. That is, the algorithm transfers item $o$ from $\pi^i_1$ to $\pi^i_2$ and from $\pi^j_2$ to $\pi^j_1$ if it is a good (Line~\ref{line:first1:swap}); it transfers item $o$ from $\pi^i_2$ to $\pi^i_1$ and from $\pi^j_1$ to $\pi^j_2$ if it is a chore (Line~\ref{line:first2:swap}). In the second phase, we swap the roles of the two agents and apply the same procedures. 

The following lemmas ensure that a beneficial transfer is possible between the envious round $\pi^j$ and the envy-free round $\pi^i$ while keeping the PO and EF property.

\begin{lemma}\label{lem:benefinical-swap}
Let $\reppi{k}$ be a k-round allocation. 
Take any pair of distinct rounds $i,j \in [k]$. Suppose that agent~$1$ envies agent~$2$ at $\pi^j$ but does not envy agent~$2$ at $\pi^i$. Then, there exits an item $o$ such that
\begin{itemize}
    \item $o \in \pi^i_1 \setminus \pi^j_1$ and $u_1(o) >0$, or
    \item $o \in \pi^j_1 \setminus \pi^i_1$ and $u_1(o) < 0$. 
\end{itemize}
\end{lemma}
\begin{proof}
Since agent~$1$ envies agent~$2$ in the $j$-th round and does not envy in the $i$-th round, we have $u_1(\pi^{j}_1) < u_1(\pi^{j}_2)$ and $u_1(\pi^{i}_1) \geq u_1(\pi^{i}_2)$. Thus, $u_1(\pi^{i}_1)> \frac{u_1(\I)}{2} >u_1(\pi^{j}_1)$, which means that there exits an item $o \in \pi^i_1 \setminus \pi^j_1$ with $u_1(o) >0$, or there exists an item $o \in \pi^j_1 \setminus \pi^i_1$ with $u_1(o) < 0$. 
\end{proof}

\begin{lemma}\label{lem:swap}
Suppose $n=2$ and $\N=\{1,2\}$.
Let $\reppi{k}$ be a k-round allocation that is PO and EF. 
Take any pair of distinct rounds $i,j \in [k]$. Consider the following $k$-round allocation $({\hat \pi}^1,\ldots,{\hat \pi}^k)$ where ${\hat \pi}_t=\pi_t$ for $t \neq i,j$ and there exists $o \in {\hat \pi}^i_1$ such that
\begin{itemize}
    \item ${\hat \pi}^i_1=\pi^i_1 \setminus \{o\}$ and ${\hat \pi}^i_2=\pi^i_2 \cup \{o\}$, and
    \item ${\hat \pi}^j_1=\pi^j_1 \cup \{o\}$, and ${\hat \pi}^j_2=\pi^j_2 \setminus \{o\}$.
\end{itemize}
Then, $({\hat \pi}^1,\ldots,{\hat \pi}^k)$ is both PO and EF overall. 
\end{lemma}
\begin{proof}
The claim clearly holds since the two $k$-round allocations $\reppi{k}$ and $({\hat \pi}^1,\ldots,{\hat \pi}^k)$ assign the same set of items to each agent. 
\end{proof}

\begin{algorithm}
\caption{Algorithm for computing a PO and EF allocation that is per-round weak EF1 }
\label{alg:POEF1wEF1}
\SetInd{0.8em}{0.3em}
{\bf Input}: A PO and EF allocation ${\bar \pi}^{(k)}=({\pi}^1,\ldots,{\pi}^k)$\;
Let $E_1$ denote the set of indices $i$ such that $\pi^i$ is not envy-free for agent~$1$. Let $F_1=[k] \setminus E_1$.\;
\tcp*[h]{Envy-adjustment phase for agent~$1$}\\
\While{there exists $j \in E_1$ such that $\pi^j$ is not weak EF1 for agent~$1$\label{line:firstwhile}}{
Take such $j \in E_1$\label{line:first:envy}\;
\While{$\pi^j$ is not weak EF1 for agent~$1$\label{line:firstwhile2}}{
Let $i$ be a round with minimum index $i \in F_1$\label{line:first:EF}\;
Take item $o$ such that $o \in \pi^{i}_1 \setminus \pi^{j}_1$ with $u_1(o)>0$, or $o \in \pi^{j}_1 \setminus \pi^{i}_1$ with $u_1(o)<0$ (the existence of $o$ is guaranteed by Lemma~\ref{lem:benefinical-swap})\label{line:first:if}\;
\If{$u_1(o)>0$}{
Set ${\pi}^{i}_1=\pi^{i}_1 \setminus \{o\}$, ${\pi}^{i}_2=\pi^{i}_2 \cup \{o\}$, ${\pi}^j_1=\pi^j_1 \cup \{o\}$, and ${\pi}^j_2=\pi^j_2 \setminus \{o\}$\label{line:first1:swap}\;
}
\Else{
Set ${\pi}^{j}_1=\pi^{j}_1 \setminus \{o\}$, ${\pi}^{j}_2=\pi^{j}_2 \cup \{o\}$, ${\pi}^i_1=\pi^i_1 \cup \{o\}$, and ${\pi}^i_2=\pi^i_2 \setminus \{o\}$\label{line:first2:swap}\;
}

\If{$\pi^{i}$ is not envy-free for agent~$1$}{
Set $F_1=F_1 \setminus \{i\}$ and $E_1= E_1 \cup \{i\}$ \label{line:firstwhileend}\;
}
}
}
\tcp*[h]{Envy-adjustment phase for agent~$2$}\\
Let $E_2$ denote the set of indices $i$ such that $\pi^i$ is not envy-free for agent~$2$. (Note that $E_2 \subseteq F_1$ at this point.) Let $F_2=[k] \setminus E_2$.\;
\While{there exists $j \in E_2$ such that $\pi^j$ is not weak EF1 for agent~$2$\label{line:secondwhile}}{
Take such $j \in E_2$\label{line:second:envy}\;
\While{$\pi^j$ is not weak EF1 for agent~$2$}{
Let $i$ be a round with minimum index $i \in F_2$\label{line:second:EF}\;
Take item $o$ such that $o \in \pi^{i}_2 \setminus \pi^{j}_2$ with $u_2(o)>0$, or $o \in \pi^{j}_2 \setminus \pi^{i}_2$ with $u_2(o)<0$ (the existence of $o$ is guaranteed by Lemma~\ref{lem:benefinical-swap})\label{line:second:if}\;
\If{$u_2(o)>0$}{
Set ${\pi}^{i}_2=\pi^{i}_2 \setminus \{o\}$, ${\pi}^{i}_1=\pi^{i}_1 \cup \{o\}$, ${\pi}^j_2=\pi^j_2 \cup \{o\}$, and ${\pi}^j_1=\pi^j_1 \setminus \{o\}$\label{line:second1:swap}\;
}
\Else{
Set ${\pi}^{j}_2=\pi^{j}_2 \setminus \{o\}$, ${\pi}^{j}_1=\pi^{j}_1 \cup \{o\}$, ${\pi}^i_2=\pi^i_2 \cup \{o\}$, and ${\pi}^i_1=\pi^i_1 \setminus \{o\}$\label{line:second2:swap}\;
}
\If{$\pi^{i}$ is not envy-free for agent~$2$}{
Set $F_2=F_2 \setminus \{i\}$ and $E_2= E_2 \cup \{i\}$\label{line:secondwhileend}\;
}
}
}
\end{algorithm}

We are now ready to prove Theorem~\ref{thm:POEFperroundwEF1}. 

\begin{proof}[Proof of Theorem~\ref{thm:POEFperroundwEF1}]
In the following, we refer to the first {\bf while} loop (Lines \ref{line:firstwhile} -- \ref{line:firstwhileend}) as the \emph{first phase} and the second {\bf while} loop (Lines \ref{line:secondwhile} -- \ref{line:secondwhileend}) as the \emph{second phase}.  

To see that the first phase is well-defined, we show that during the execution of the first phase, $F_1$ is nonempty and corresponds to the rounds where agent~$1$ is envy-free. Indeed, consider the allocation $\pi^j$ just after Line~\ref{line:first2:swap}; since $\pi^j$ is not weak EF1 before the transfer operation in Lines~\ref{line:first1:swap} and~\ref{line:first2:swap}, agent~$1$ remains envious at $\pi^j$ just after Line~\ref{line:first2:swap}, and hence $j$ continues to belong to $E_1$. Thus, during the execution of the first phase, $E_1$ corresponds to the set of the rounds where agent~$1$ envies the other while $F_1$ corresponds to the set of rounds where agent~$1$ is envy-free. 
Further, by Lemma~\ref{lem:swap}, the algorithm keeps the property that the allocation $\reppi{k}$ is PO and EF; by Pareto-optimality of $\reppi{k}$, each $\pi^i$ for $i \in [k]$ is Pareto-optimal within the round. Thus, at most one agent envies the other at each $\pi^i$. 
Therefore, during the execution of the first phase, $\reppi{k}$ remains envy-free for agent~$1$, which implies that there exists a round $i$ such that $\pi^i$ is envy-free for agent~$1$, namely, $F_1$ is nonempty. By Lemma~\ref{lem:benefinical-swap}, there exists an item $o$ satisfying the condition in Line~\ref{line:first:if}. A similar argument shows that the second phase is also well-defined.

Next, we show that the first phase terminates in polynomial time and the allocation $\reppi{k}$ just after the first phase is per-round weak EF1 for agent~$1$. To see this, consider $\pi^i$ defined in Line~\ref{line:first:EF}. After the transfer operation in Lines~\ref{line:first1:swap} and~\ref{line:first2:swap}, $\pi^i$ does not violate weak EF1 from the viewpoint of agent~$1$; indeed, agent~$1$ may envy the other at $\pi^i$ after the transfer operation, but since $\pi^i$ is envy-free for agent~$1$ before, the envy can be eliminated by either stealing the good $o$ $(u_1(o)>0)$ from the other agent, or transferring the chore $o$ $(u_1(o)<0)$ to the other agent. 
Now, consider the allocation $\pi^j$ that is not weak EF1 for agent~$1$. Clearly, at each iteration of the {\bf while} loop in Lines \ref{line:firstwhile2} -- \ref{line:firstwhileend}, the number of goods $o$ $(u_1(o)>0)$ increases while the number of chores $o$ $(u_1(o)<0)$ decreases for agent~$1$. Thus, $\pi^j$ becomes weak EF1 for agent~$1$ in $O(m)$ iterations. 
Thus, we conclude that the first phase terminates in polynomial time and the allocation $\reppi{k}$ just after the first phase is per-round weak EF1 for agent~$1$. 

A similar argument shows that the second phase terminates in polynomial time and the final allocation is per-round weak EF1 for agent~$2$. It remains to show that $\reppi{k}$ remains weak EF1 for agent~$1$ during the second phase. 

Consider an arbitrary iteration in the second phase and allocations $\pi^i$ and $\pi^j$ defined in Lines~\ref{line:second:envy} and \ref{line:second:EF}, respectively. Let $o$ be an item chosen in Line~\ref{line:second:if}. Assume that $\pi^i$ and $\pi^j$ are weak EF1 for agent~$1$ before the transfer operation in Lines~\ref{line:second1:swap} and~\ref{line:second2:swap}. We show that both $\pi^i$ and $\pi^j$ remain weak EF1 for agent~$1$ just after the transfer operation.  
First, consider $\pi^i$. Since after the swap $\pi^i$ remains Pareto-optimal within the round, we have: 
\begin{itemize}
    \item If $u_2(o)>0$, then $u_1(o)>0$. 
    \item If $u_2(o)<0$, then $u_1(o)<0$. 
\end{itemize}
Thus, agent~$1$'s utility does not decrease after the transfer operation. Thus, $\pi^i$ remains weak EF1 for agent~$1$. 

Next, consider $\pi^j$. As we have observed before, agent~$2$ remains envious at $\pi^j$ just after the transfer operation in Lines~\ref{line:second1:swap} and~\ref{line:second2:swap} since $\pi^j$ is not weak EF1 for agent~$2$ before the transfer operation. This means that by Pareto-optimality of $\pi$, $\pi^j$ is still envy-free for agent~$1$ after the transfer operation. Thus, $\pi^j$ is weak EF1 for agent~$1$. 

It is not difficult to see that Algorithm~\ref{alg:POEF1wEF1} runs in polynomial time. 
\end{proof}

Finally, we show that, if we do not require PO, an allocation that is EF and per-round EF1 can always be found (when $k$ is even) in polynomial time.

%\begin{theorem}
\begin{restatable}{theorem}{TheoremEF}\label{thm:2-agents-EF+per-round EF1}
If $n=2$ and $k\in 2\mathbb{N}$, then an allocation which is envy-free overall and per-round EF1 always exists, and can be computed in polynomial time.
\end{restatable}
\begin{proof}
Consider the following procedure. Create a sequence of items $P=(o_1,o_2,\ldots,o_m)$. For $j=1,2,\ldots,m$, let $I_j=\{o_1,o_2,\ldots,o_j\}$ and ${\bar I}_j=\{o_{j+1},o_{j+2},\ldots,o_m\}$. Let $I_0=\emptyset$ and ${\bar I}_0=I_m$. Observe that we have either
\begin{itemize}
\item $u_1(I_0) \leq u_1({\bar I}_0)$ and $u_1(I_m) \geq u_1({\bar I}_m)$, or
\item $u_1(I_0) \geq u_1({\bar I}_0)$ and $u_1(I_m) \leq u_1({\bar I}_m)$. 
\end{itemize}
Thus, there exists an index $j \in \{1,2,\ldots,m\}$ where the preference of agent~$1$ switches when $o_j$ moves from the right to the left bundle, namely, 
\begin{itemize}
\item[(\emph{i})] $u_1(I_{j-1}) \leq u_1({\bar I}_{j-1})$ and $u_1(I_{j}) \geq u_1({\bar I}_{j})$, or
\item[(\emph{ii})] $u_1(I_{j-1}) \geq u_1({\bar I}_{j-1})$ and $u_1(I_{j}) \leq u_1({\bar I}_{j})$. 
\end{itemize}

Let $L=I_{j-1}$ and $R={\bar I}_{j}$. Assume without loss of generality that $u_1(L) \geq u_1(R)$. 

First, suppose that we are in the first case (\emph{i}), i.e., $u_1(I_{j-1}) \leq u_1({\bar I}_{j-1})$ and $u_1(I_{j}) \geq u_1({\bar I}_{j})$. In other words, this means that agent~$1$ weakly prefer a bundle with $o_j$: $u_1(L) \leq u_1(R\cup \{o_j\})$ and $u_1(L\cup \{o_j\}) \geq u_1(R)$.  
Consider the following cases. 
\begin{itemize} 
    \item Suppose that agent~$2$ weakly prefers $L$ to the remaining items, or $R$ to the remaining items, i.e., $u_2(L) \geq u_2(R\cup \{o_j\})$ or $u_2(R) \geq u_2(L\cup \{o_j\})$. Then, there is an envy-free allocation and hence the sequence that repeats this allocation $k$ times is a desired solution. Indeed, if $u_2(L) \geq u_2(R\cup \{o_j\})$, then the allocation that allocates $R$ together with $o_j$ to agent~$1$ and $L$ to agent~$2$ is an EF allocation. If $u_2(R) \geq u_2(L\cup \{o_j\})$, then the allocation that allocates $L$ together with $o_j$ to agent~$1$ and $R$ to agent~$2$  is an EF allocation. 

    \item Suppose that agent~$2$ weakly prefers $R \cup \{o_j\}$ to $L$, and $L \cup \{o_j\}$ to $R$, namely, $u_2(R\cup \{o_j\}) \geq u_2(L)$ and $u_2(L\cup \{o_j\}) \geq  u_2(R)$. If $u_2(R) \leq u_2(L)$, this means that both agents weakly prefer $L$ to $R$; thus, a sequence that repeatedly swaps two bundles $L$ and $R\cup \{o_j\}$ among the two agents is a per-round EF1 and EF. 
    %, then repeat two allocations $(L,R\cup \{o_j\})$ and $(R\cup \{o_j\},L)$. Such a sequence is per-round EF1 and EF. 
    On the other hand, if $u_2(L) \leq  u_2(R)$, this means that the two agents have a different preference among $L$ and $R$: while agent~$1$ weakly prefers $L$ to $R$, agent~$2$ weakly prefers $R$ to $L$. Thus, allocate $L$ to agent~$1$ and $R$ to agent~$2$ and alternate between assigning item $o_j$ to agent~$1$ and to agent~$2$. 
    %then repeat two allocations $(L,R\cup \{o_j\})$ and $(L\cup \{o_j\},R)$. 
    Clearly, such a sequence is per-round EF1. To see that it is EF, observe that we have $u_1(L\cup \{o_j\}) \geq u_1(R\cup \{o_j\})$ and $u_1(L) \geq u_1(R)$ and thus agent~$1$ does not envy agent~$2$ at the $k$-round allocation. Similarly, we have $u_2(R\cup \{o_j\}) \geq u_2(L\cup \{o_j\})$ and $u_2(R) \geq u_2(L)$ and thus agent~$2$ does not envy agent~$1$ at the $k$-round allocation. 
\end{itemize}

Next, suppose that we are in the second case (\emph{ii}), i.e., $u_1(I_{j-1}) \geq u_1({\bar I}_{j-1})$ and $u_1(I_{j}) \leq u_1({\bar I}_{j})$. In other words, agent~$1$ weakly prefer a bundle without $o_j$: $u_1(L) \geq u_1(R\cup \{o_j\})$ and $u_1(L\cup \{o_j\}) \leq u_1(R)$. 
Consider the following cases. 
\begin{itemize} 
    \item Suppose that agent~$2$ weakly prefers $R \cup \{o_j\}$ to $L$, or $L \cup \{o_j\}$ to $R$, namely, $u_2(R\cup \{o_j\}) \geq u_2(L)$ or $u_2(L\cup \{o_j\}) \geq  u_2(R)$. Then, we show that there is an envy-free allocation and hence the sequence that repeats this allocation $k$ times is a desired solution.  
    Indeed, if $u_2(R\cup \{o_j\}) \geq u_2(L)$, then the allocation giving $L$ to agent $1$ and $R\cup \{o_j\})$ to agent $2$ is an EF allocation. 
    Similarly, if $u_2(L\cup \{o_j\}) \geq  u_2(R)$, then the allocation giving $R$ to agent $1$ and $L\cup \{o_j\})$ to agent $2$ is an EF allocation.

    \item Suppose that agent~$2$ weakly prefers $L$ to $R \cup \{o_j\}$, and $R$ to $L \cup \{o_j\}$. 
    If $u_2(L) \geq  u_2(R)$, this means that both agents weakly prefer $L$ to $R$; thus, a sequence that repeatedly swaps two bundles $L\cup \{o_j\}$ and $R$ among the two agents is a per-round EF1 and EF. 
    
    If $u_2(R) \geq u_2(L)$, then this means that the two agents have a different preference among $L$ and $R$: while agent~$1$ weakly prefers $L$ to $R$, agent~$2$ weakly prefers $R$ to $L$. Thus, allocate $L$ to agent~$1$ and $R$ to agent~$2$ and alternate between assigning item $o_j$ to agent~$1$ and to agent~$2$. Such a sequence is per-round EF1. To see that it is EF, observe that we have $u_1(L) \geq u_1(R)$ and $u_1(L\cup \{o_j\}) \geq u_1(R\cup \{o_j\})$ and thus agent~$1$ does not envy agent~$2$. Similarly, since $u_2(R\cup \{o_j\}) \geq u_2(L\cup \{o_j\})$ and $u_2(R) \geq u_2(L)$, agent~$2$ does not envy agent~$1$.
\end{itemize}

It is immediate to see that the above procedure can be implemented in polynomial time. 
This concludes the proof.
\end{proof}

We do not know whether EF and per-round EF1 can be simultaneously achieved for $n>2$ agents and $k\in n \mathbb{N}$ rounds (EF overall is possible in this case by Proposition~\ref{prop:EF-with-cn-rounds}). We leave it as an interesting open problem for future work.

\section{Variable number of rounds}\label{sec:variable}

In the previous sections, we have assumed that $k$ (the number of rounds) is predetermined. In this section, we study the case when $k$ can be \emph{variable}.

We have seen that, whenever $k$ is fixed, overall EF and PO become incompatible in general (Theorem~\ref{thm:no-EF-PO-with-n-greater-than-2}). This raises the question of whether this is possible if $k$ is not predetermined. More precisely, we ask the following: given a utility profile, is there a number $k$ of rounds for which there is a sequence of allocations $\reppi{k}$ that satisfies desirable overall and per-round guarantees? We answer this question affirmatively.

\begin{theorem}\label{thm:variable:PO+EF+PROP2}
    If for all $i\in\N$ and $o\in\I$ we have $u_i(o)\in\mathbb{Q}$, then there exists $k\in\mathbb{N}$ and a $\reppi{k}$ such that $\reppi{k}$ is envy-free, Pareto-optimal, and per-round PROP[1,1].\label{thm:EF-PO-n-agents-fractional}
\end{theorem}

Recall that PROP[1,1] and PROP1 coincide for the case of non-negative value items and for the case of non-positive value items; thus, in such cases, the above theorem holds with a per-round guarantee of PROP1. 

To show Theorem~\ref{thm:variable:PO+EF+PROP2}, we establish a connection between our setting and the divisible~\citep{BMSY17a} and probabilistic settings~\citep{aziz2023best,budish2013designing}. More precisely, we utilize a proof technique similar to that of \citet{aziz2023best}, who used a decomposition lemma by \citet{budish2013designing} for divisible item allocations to construct a randomized allocation that satisfies desirable efficiency and fairness notions when all items are goods. 
We prove that there exists a randomized allocation satisfying similar desirable guarantees for the mixed case of goods and chores and translate it into our repeated setting. 

Before we proceed, we provide some preliminary definitions. Consider a set of $n$ agents $\N=[n]$ and of items $\I$. 
%Define randomized allocations 
A \emph{randomized allocation} is a set of ordered pairs
$(p^t,\pi^{t})_{t \in [h]}$, such that for every $t \in [h]$, $\pi^{t}$ is an allocation implemented with probability
$p^t \in [0,1]$, where $\sum_{t \in [h]}p_t=1$; the allocations $\pi^{1}, \pi^{2}, \ldots,\pi^{h}$ are referred to as the \emph{support} of a randomized allocation.

Given a randomized allocation $(p^t,\pi^{t})_{t \in [h]}$ with $p^t \in \mathbb{Q}$ for each $t \in [h]$, we define its \emph{repeated translation} as follows. Since all $p^t$ are rational numbers, there exists a $k\in\mathbb{N}$ such that, for all $t \in [h]$, $p_t=\nicefrac{\ell_t}{k}$ (for some $\ell_t\in\mathbb{N}$). In other words, all fractions defined by $(p_t)_{t \in [h]}$ are expressed in terms of the same denominator, $k$. 
The repeated translation of $(p^t,\pi^{t})_{t \in [h]}$ is the allocation sequence $\reppi{k}$ of form:
\[
    \reppi{k}=(\underbrace{\pi^1,\ \ldots,\ \pi^1}_{\ell_1\text{ times}},\ \underbrace{\pi^2,\ldots,\ \pi^2}_{\ell_2\text{ times}},\ \ldots,\ \underbrace{\pi^h,\ \ldots,\ \pi^h}_{\ell_h\text{ times}}).
\]
In other words, for each $t \in [h]$, $\pi^{t}$ appears exactly $\ell_t$ times (where $p_t=\nicefrac{\ell_t}{k}$).

We first show the following lemma, stating that any PO and EF fractional allocation can be decomposed into a randomized allocation whose support satisfy PROP[1,1].
\begin{lemma}\label{lem:decompostion:POEF}
Suppose that $\bfx$ is a PO and EF fractional allocation. Then there exists a randomized allocation $(p^t,\pi^{t})_{t \in [h]}$ that implements $\bfx$ with the following properties: 
\begin{enumerate}
\item $p_1,p_2,\ldots,p_{h} \in \mathbb{Q}_+$,  
\item $\pi^{t}_i \subseteq \{\, o \in I \mid x_{i,o}>0\,\}$ for each $i \in N$ and $t \in [h]$, and 
\item $\pi^{t}$ is PROP[1,1] for each $t \in [h]$. 
\end{enumerate}    
\end{lemma}

To prove the above lemma, we utilize Lemma~\ref{lem:decompostion:budis}, proven by \citet{budish2013designing}.\footnote{In fact, \citet{budish2013designing} showed a more general statement of Lemma~\ref{lem:decompostion:budis} for \emph{bihierarchical} constraints $\mathcal{H}$ where $\mathcal{H}$ is a disjoint union of $\mathcal{H}_1,\mathcal{H}_2 \subseteq 2^{N \times I}$ such that each $\mathcal{H}_j$ forms a laminar set family and one of $\mathcal{H}_j$ requires a fractional allocation to be valid. } 
%Preliminary definitions
Specifically, consider a family $\mathcal{H} \subseteq 2^{N \times I}$ of subsets of agent-item pairs. 
We assume that $\mathcal{H}$ forms a \emph{laminar set family}, namely, $S, S' \in \mathcal{H}$ implies either $S \subseteq S'$, $S' \subseteq S$, or $S \cap S' =\emptyset$. 
Given a fractional allocation $\bfx$, we write $x_S=\sum_{(i,o) \in S}x_{i,o}$ for each $S \in \mathcal{H}$ and let $\underline{q}_S, \overline{q}_S$ be non-negative real numbers for each $S \in \mathcal{H}$. Consider the following constraints on fractional allocations $\bfx$ with lower and upper quotas on the amount of fractions that can be allocated to the agent-item pairs in $S$: 
\begin{equation}\label{eq:H}
\underline{q}_S \leq x_S \leq \overline{q}_S~\mbox{\quad for each}~S \in \mathcal{H}.
\end{equation}

\begin{lemma}[\citet{budish2013designing}]\label{lem:decompostion:budis}
Suppose that $\mathcal{H}$ is a laminar set family and $\bfx^*$ is a fractional allocation  satisfying constraints~\eqref{eq:H}. Then, one can find in strongly polynomial-time a randomized allocation $(p^t,\pi^{t})_{t \in [h]}$ that implements $\bfx^*$ such that each allocation $\pi^{t}$ for $t \in [h]$ in the support satisfies~\eqref{eq:H}, namely:
\begin{equation*}
\underline{q}_S \leq \sum_{(i,o) \in S} \boldsymbol{1}[o \in \pi^t_i] \leq \overline{q}_S~\mbox{\quad for each}~S \in \mathcal{H}.
\end{equation*}
\end{lemma}

Building up on this decomposition lemma, we show Lemma~\ref{lem:decompostion:POEF} in a similar manner to the proof of \citet{aziz2023best}. 

\begin{proof}[Proof of Lemma~\ref{lem:decompostion:POEF}]
Given a fractional allocation $\bfx^*$ satisfying PO and EF, let us define $\mathcal{H}$ as follows. 
We divide the items in $I$ into goods and chores for each agent. Specifically, for each agent $i \in I$, we let $G_i=\{\, o \in I\mid u_i(o) \geq 0\,\}$ and $C_i=\{\, o \in I\mid u_i(o) < 0\,\}$; set $m^+_i=|G_i|$ and $m^-_i=|C_i|$. 
We label the items in $G_i=\{g_{i,1},g_{i,2},\ldots,g_{i,m^+_i}\}$ so that each $g_{i,\ell}$ represents the $\ell$-th favorite subjective good for $i$, i.e., 
\[
u_i(g_{i,1}) \geq  u_i(g_{i,2}) \geq \cdots \geq u_i(g_{i,m^+_i}). 
\]
Similarly, we let $C_i=\{c_{i,1},c_{i,2},\ldots,c_{i,m^-_i}\}$ so that each $c_{i,\ell}$ represents the $\ell$-th least favorite subjective chore for $i$, i.e., 
\[
u_i(c_{i,1}) \leq  u_i(c_{i,2}) \leq \cdots \leq u_i(c_{i,m^-_i}). 
\]

For agent $i \in N$, each $\ell \in [m^+_i]$, we set $G_{i,\ell}=\{\, (i,g_{i,h}) \mid h \leq \ell \,\}$ that corresponds to the set of the first $\ell$-most preferred goods of agent $i$. Similarly, we let $C_{i,\ell}=\{\, (i,c_{i,h}) \mid h \leq \ell \,\}$ that corresponds to the set of the first $\ell$-least preferred chores of agent $i$. 
We define 
\begin{align*}
\mathcal{H}= &\{\, \{(i,o)\}\mid i \in N, o \in I\,\} \\
&\cup \{\, G_{i,\ell^+}, C_{i,\ell^-} \mid i \in N, \ell^+ \in [m^+_i], \ell^- \in [m^-_i]\,\}. 
\end{align*}
For each $S \in \mathcal{H}$, we set $\underline{q}_S=\lfloor x^*_S \rfloor$ and $\overline{q}_S=\lceil x^*_S \rceil$. 
The constraints imposed by the singleton sets $\{(i,o)\}$ require that each agent $i$ gets the full of $o$ (respectively, a zero-fraction of $o$) if $i$ gets item $o$ fully (respectively, none of item $o$) in $\bfx^*$. 
The constraints imposed by $G_{i,\ell}$ require that agent $i$ gets at least $\lfloor x^*_{G_{i,\ell}} \rfloor$ goods from the first $\ell$-most preferred goods; on the other hand, the constraints by $C_{i,\ell}$ require that agent $i$ gets at most $\lceil x^*_{C_{i,\ell}} \rceil$ chores from the first $\ell$-least preferred chores.

It is not difficult to verify that $\mathcal{H}$ is laminar. Hence, we can apply Lemma~\ref{lem:decompostion:budis} and obtain a randomized allocation $(p^t,\pi^{t})_{t \in [h]}$ that implements $\bfx^*$ and satisfies for each $t \in [h]$,
\begin{equation*}
\lfloor x^*_S \rfloor \leq \sum_{(i,o) \in S} \boldsymbol{1}[o \in \pi^t_i] \leq \lceil x^*_S \rceil~\mbox{\quad for each}~S \in \mathcal{H}.
\end{equation*}
Here, we can ensure the following properties: 
\begin{enumerate}
\item $p_1,p_2,\ldots,p_{h} \in \mathbb{Q}_+$, and
\item $\pi^{t}_i \subseteq \{\, o \in I \mid x^*_{i,o}>0\,\}$ for each $i \in N$ and $t \in [h]$, 
\end{enumerate}   
where the first condition holds because the randomized allocation can be obtained in strongly polynomial time and the second condition holds due to the singleton constraints in $\mathcal{H}$ that each item cannot be allocated to an agent $i$ who gets none of the item under the fractional allocation. 
It remains to show that $\pi^{t}$ is PROP[1,1] for each $t \in [h]$.

Now fix any agent $i \in N$ and $t \in [h]$. 
We define $U^+_i=\sum_{g \in G_i} x^*_{i,g}u_i(g)$ and $U^-_i=\sum_{c \in C_i} x^*_{i,c}u_i(c)$. Here, $U^+_i$ represents the total utility agent $i$ receives from her subjective goods under $\bfx^*$ and $U^-_i$ that from subjective chores. Observe that $U^+_i + U^-_i \geq \nicefrac{u_i(I)}{n}$ since $\bfx^*$ is envy-free (and envy-freeness implies proportionality in this setting as well). Hence, in order to establish PROP[1,1], it suffices to show the following: 
\begin{itemize}
\item[$($i$)$] if $u_i( \pi^t_i \cap G_i) < U^+_i$, there exists an item $g \in I \setminus \pi^t_i$ such that $u_i( (\pi^t_i \cap G_i) \cup \{g\}) \geq U^+_i$; and
\item[$($ii$)$] if $u_i( \pi^t_i \cap C_i) < U^-_i$, there exists an item $c \in \pi^t_i$ such that $u_i( (\pi^t_i \cap C_i) \setminus \{c\}) \geq U^-_i$. 
\end{itemize}

%\ayumi{thanks fixed!}
%\noteon{should "each" be "ease"?} 
For ease of presentation, define a matrix $\bfy=(\bfy_1,\bfy_2,\ldots,\bfy_n) \in \{0,1\}^{N \times I}$ that corresponds to an allocation $\pi^t$, where we set $y_{i,o}=1$ if $o \in \pi^t$ and $y_{i,o}=0$ otherwise. 
Recall that for each $S \in \mathcal{H}$, $y_S=\sum_{(i,o) \in S}y_{i,o}$ represents the number of agent-item pairs $(i,o)$ such that $i$ is allocated to $o$. Since $\bfy$ satisfies \eqref{eq:H}, we have that  $\lfloor x^*_{G_{i,\ell}} \rfloor -y_{G_{i,\ell}} \leq 0$ and $ y_{G_{i,\ell}}  -  \lceil  x^*_{G_{i,\ell}} \rceil\leq 0$ for every $\ell \in [m^+_i]$.

%Utility guarantee for goods
First, we show $($i$)$. Suppose $u_i( \pi^t_i \cap G_i) < U^+_i$. Recall that agents can only get items for which they receive a positive fraction in $\bfx^*$, meaning that $\pi^t_i \cap G_i \subseteq \{\, o \in G_i \mid x^*_{i,o}>0\,\}$. Thus, this inclusion relation must be proper, i.e., there is an item $g_{i,\ell^+}$ such that $x^*_{i,g_{i,\ell^+}}>0$ but $g_{i,\ell^+} \not \in \pi^t_i$. 
Define $\ell^+ \in [m^+_i]$ to be the smallest index of such an item in $G_i$. 
We create a dummy item $g_{i,m^+_i+1}$ where $u_i(g_{i,m^+_i+1})=0$.
Then, 
\begin{align*}
&U^+_i - u_i( \pi^t_i \cap G_i)\\
&=\sum^{m^+_i}_{\ell=1}  u_i(g_{i,\ell}) (x^*_{i,g_{i,\ell}} - y_{i,g_{i,\ell}})\\
&=\sum^{m^+_i}_{\ell=1}  [u_i(g_{i,\ell}) + \sum^{m^+_i}_{q=\ell+1} (u_i(g_{i,q}) - u_i(g_{i,q}))] (x^*_{i,g_{i,\ell}} - y_{i,g_{i,\ell}})\\
&= \sum^{m^+_i}_{\ell=1} (u_i(g_{i,\ell}) - u_i(g_{i,\ell+1}) ) (x^*_{G_{i,\ell}}-y_{G_{i,\ell}})\\
&\leq \sum^{m^+_i}_{\ell=\ell^+} (u_i(g_{i,\ell}) - u_i(g_{i,\ell+1}) ) (x^*_{G_{i,\ell}}-y_{G_{i,\ell}}) \\
&\leq \sum^{m^+_i}_{\ell=\ell^+} (u_i(g_{i,\ell}) - u_i(g_{i,\ell+1}) ) \cdot 1 \\
&\leq u_i(g_{i,\ell^+}). 
\end{align*}
The third transition follows using simple algebra.\footnote{For example, when $m^+_i=3$,
we have 
$
(u(1)-u(2)+u(2)-u(3)+u(3))(x^*_{1}-y_{1})
+(u(2)-u(3)+u(3)) (x^*_{2} - y_{2})+ u(3) (x^*_{3}-y_{3})
= (u(1)-u(2))(x^*_{1}-y_{1})
+ (u(2)-u(3))(x^*_{1}+x^*_2 -y_{1}-y_{2}) 
+ (u(3)-0 )(x^*_{1}+x^*_{2}+x^*_{3} -y_{1}-y_{2}-y_{3})$,
where we write $u(\ell)=u_i(g_{i,\ell})$, $x^*_{\ell}=x^*_{i,g_{i,\ell}}$, and $y_{\ell}=y_{i,g_{i,\ell}}$ for each $\ell=1,2,3$.}
The fourth transition holds since $u_i(g_{i,\ell}) - u_i(g_{i,\ell+1}) \geq 0$ for every $\ell \in [m^+_i]$ and $x^*_{G_{i,\ell}}-y_{G_{i,\ell}}=x^*_{G_{i,\ell}}-|G_{i,\ell}| \leq 0$ for every $\ell \leq \ell^+-1$. The fith transition holds since $x^*_{G_{i,\ell}}-y_{G_{i,\ell}} \leq \lfloor x^*_{G_{i,\ell}} \rfloor +1 -y_{G_{i,\ell}} \leq 1$ for every $\ell \in [m^+_i]$. Hence, $($i$)$ holds.

%Utility guarantee for chores
Next, we show $($ii$)$. Suppose $u_i( \pi^t_i \cap C_i) < U^-_i$. Then, there is an item $c_{i,\ell^-}$ such that $x^*_{i,g_{i,\ell^-}}>0$ and $c_{i,\ell^-} \in \pi^t_i$. Define $\ell^- \in [m^-_i]$ to be the smallest index of such an item in $C_i$. Again, we create a dummy item $c_{i,m^-_i+1}$ where $u_i(c_{i,m^-_i+1})=0$. 
Then,

\begin{align*}
&u_i( \pi^t_i \cap C_i)-U^-_i\\
&=\sum^{m^-_i}_{\ell=1}  u_i(c_{i,\ell}) (y_{i,c_{i,\ell}}-x^*_{i,c_{i,\ell}})\\
&=\sum^{m^-_i}_{\ell=1}  [u_i(c_{i,\ell}) + \sum^{m^-_i}_{q=\ell+1} (u_i(c_{i,q}) - u_i(c_{i,q}))] (y_{i,c_{i,\ell}}-x^*_{i,c_{i,\ell}})\\
&= \sum^{m^-_i}_{\ell=1} (u_i(c_{i,\ell}) - u_i(c_{i,\ell+1}) ) (y_{C_{i,\ell}}-x^*_{C_{i,\ell}})\\
&\geq \sum^{m^-_i}_{\ell=\ell^-} (u_i(c_{i,\ell}) - u_i(c_{i,\ell+1}) )  (y_{C_{i,\ell}}-x^*_{C_{i,\ell}}) \\
&\geq \sum^{m^-_i}_{\ell=\ell^-} (u_i(c_{i,\ell}) - u_i(c_{i,\ell+1}) ) \cdot 1 \\
&\geq u_i(c_{i,\ell^-}). 
\end{align*}

The fourth transition holds since $u_i(c_{i,\ell}) - u_i(c_{i,\ell+1}) \leq 0$ for every $\ell \in [m^-_i]$ and $y_{C_{i,\ell}}-x^*_{C_{i,\ell}} = - x^*_{C_{i,\ell}} \leq 0$ for every $\ell \leq \ell^- -1$. The fifth transition holds since $y_{G_{i,\ell}}-x^*_{G_{i,\ell}} \leq y_{G_{i,\ell}}  -  (\lceil  x^*_{G_{i,\ell}} \rceil -1) \leq 1$ for every $\ell \in [m^+_i]$. Hence, $($ii$)$ holds too.

Finally, if $u_i( \pi^t_i) \geq U^+_i + U^-_i$, we have that $u_i( \pi^t_i) \geq \nicefrac{\sat_i(\I)}{n}$. 
If, on the other hand, $u_i( \pi^t_i) < U^+_i + U^-_i$, we have either 
\begin{itemize}
\item $u_i( \pi^t_i \cap G_i) < U^+_i$ and $u_i( \pi^t_i \cap C_i) < U^-_i$, 
\item $u_i( \pi^t_i \cap G_i) \geq U^+_i$ and $u_i( \pi^t_i \cap C_i) < U^-_i$, or 
\item $u_i( \pi^t_i \cap G_i) < U^+_i$ and $u_i( \pi^t_i \cap C_i) \geq U^-_i$.
\end{itemize}
Together with $($i$)$ and $($ii$)$, this implies that $\sat_i((\pi_i \setminus X)\cup Y) \geq \nicefrac{\sat_i(\I)}{n}$ for some $X \subseteq \pi_i$ and $Y \subseteq I \setminus \pi_i$ with $|X|,|Y| \leq 1$. Therefore, PROP[1,1] holds for $\bfy$, concluding the proof.
\end{proof}

Next, we show that a randomized allocation that implements PO and EF fractional allocation can be translated into a repeated allocation with the same fairness and efficiency guarantees. 

\begin{lemma}\label{lem:fractionalPOEF}
    Let $\bfx$ be a PO and EF fractional allocation and $(p^t,\pi^{t})_{t \in [h]}$ a randomized allocation that implements $\bfx$ with $p^t \in \mathbb{Q}$ for each $t \in [h]$. Then, the repeated translation of $(p^t,\pi^{t})_{t \in [h]}$ is PO and EF.\label{lemma:translation-fractional}
\end{lemma}

\begin{proof}[Proof]
    Consider a PO and EF fractional allocation $\bfx$ and a randomized allocation $(p^t,\pi^{t})_{t \in [h]}$ that implements $\bfx$ with $p^t \in \mathbb{Q}$ for every $t \in [h]$. 
    Then, there exists $k\in\mathbb{N}$ such that for all $t \in [h]$, $p^t=\nicefrac{\ell_t}{k}$ (for some $\ell_t\in\mathbb{N}$). Let $\oreppi{k}$ be the repeated translation of $(p^t,\pi^{t})_{t \in [h]}$. Consider any agent $i\in\N$. Then, we have the following:

    \begin{align*}
        u_i(\bfx_i) &= \sum_{o\in\I}u_i(o) x_{i,o} \\
        &= \sum_{o\in\I}u_i(o) \sum^{h}_{t=1} p^t {\boldsymbol{1}}[o \in \pi^t_i]\\
        &= \sum_{o\in\I}u_i(o) \sum^{h}_{t=1} \nicefrac{\ell_t}{k} {\boldsymbol{1}}[o \in \pi^t_i]\\
        & = \nicefrac{1}{k} \sum_{o\in\I}u_i(o) \sum^{h}_{t=1} \ell_t {\boldsymbol{1}}[o \in \pi^t_i]\\
        & = \nicefrac{1}{k}\sum_{o\in\I}u_i(o) |\{r\in[k]\colon o\in\rho^r_i \}|=\nicefrac{1}{k} \cdot  u_i(\rho^{\cup k}_i).
    \end{align*}
    Therefore, the utilities yielded by $\bfx$ and $\oreppi{k}$ are the same, up to a (common for all agents) rescaling. Therefore, if $\bfx$ is EF, then so is $\oreppi{k}$. %With a similar argument, one shows that $\oreppi{k}$ must also be PO.

    Next, suppose toward a contradiction that $\bfx$ is PO, but $\oreppi{k}$ is not. Then, the latter is Pareto-dominated by some $\areppi{k}$. Consider the \emph{fractional translation} $\bfy=(\bfy_1,\bfy_2,\ldots,\bfy_n)$ of $\areppi{k}$. Namely, for every $i\in\N$ and $o\in\I$, $y_{i,o}=\nicefrac{q}{k}$ (where $q$ is the number of times $i$ receives $o$ in $\areppi{k}$). Similar to the above, the utilities yielded by $\bfy$ and $\areppi{k}$ are the same up to a (common for all agents) rescaling. Thus, the fractional allocation $\bfy$ Pareto-dominates $\bfx$, which is a contradiction. 
    \iffalse
    \[
    k\cdot u_i(\bfx_i) =\sum_{o\in\I}k\cdot x_{i,o} u_i(o)
        = u_i(\pi_i^{\cup k}),
    \]
    i.e., the utilities yielded by $\bfx$ and $\reppi{k}$ are the same, up to a (common for all agents) rescaling. Therefore, if $\bfx$ is EF, then $\reppi{k}$ also is. Next, suppose toward a contradiction that $\bfx$ is PO, but $\reppi{k}$ is not. Then, the latter is dominated by some $\areppi{k}$. But then the fractional translation of $\areppi{k}$ dominates $\bfx$, which is a contradiction. 
    %for each agent $i \in N$, the expected utility under the randomized allocation and the utility under the fractional allocation are the same, i.e., $\sum^{h}_{t=1} p_t u_i(\pi^{h}_i)=\sum_{o\in\I}x_{i,o} u_i(o)$
    \fi
\end{proof}

We are now ready to prove Theorem~\ref{thm:EF-PO-n-agents-fractional}.

\begin{proof}[Proof of Theorem~\ref{thm:EF-PO-n-agents-fractional}]
    %By Lemma~\ref{lemma:translation-fractional}, we know it is sufficient to guarantee the existence of some PO and EF fractional allocation. 
    A \emph{fractional allocation} is a vector $\bfx=(\bfx_1,\ldots,\bfx_n)$ where each $\bfx_i=(x_{i,o})_{o \in I} \in \mathbb{Q}^m_{+}$ and, for all $o\in\I$, $\sum_{i\in\N}x_{i,o}=1$. Given a utility vector $\boldsymbol{u}$, the utility of agent~$i$ for allocation $\bfx$ is defined as $u_i(\bfx_i)=\sum_{o\in\I}x_{i,o} u_i(o)$. The notions of envy-freeness and Pareto-optimality (w.r.t. all other fractional allocations) naturally translate to this setting. 
Next, a randomized allocation $(p^t,\pi^{t})_{t \in [h]}$ is said to \emph{implement} a fractional allocation $\bfx$ if
\[x_{i,o}=\sum^{h}_{t=1} p_t {\boldsymbol{1}}[o \in \pi^t_i]\qquad\text{for each }i \in N\text{ and }o \in I.\]
Here, ${\boldsymbol{1}}[o \in \pi^t_i]$ is an indicator function which takes value $1$ if $o \in \pi^h_i$, and $0$ otherwise.

When all agents' utilities are rational-valued, a PO and EF fractional allocation always exists~\citep{BMSY17a,ChaudhuryMOR}. 
    Moreover, by Lemma~\ref{lem:decompostion:POEF}, we can show that any PO and EF fractional allocation $\bfx$ admits a randomized allocation $(p^t,\pi^{t})_{t \in [h]}$ that implements $\bfx$ with these properties: 
\begin{enumerate}
\item $p_1,p_2,\ldots,p_{h} \in \mathbb{Q}_+$,  
\item $\pi^{t}_i \subseteq \{ o \in I \mid x_{i,o}>0\}$ for each $i \in N$ and $t \in [h]$, and 
\item $\pi^{t}$ is PROP[1,1] for each $t \in [h]$. 
\end{enumerate}    

Finally, by Lemma~\ref{lem:fractionalPOEF}, the repeated translation of $(p^t,\pi^{t})_{t \in [h]}$ is PO and EF overall. Moreover, since $\pi^t$ is PROP[1,1] for all $t\in[h]$, this translation is per-round PROP[1,1].
\end{proof}

Note that due to Theorem~\ref{thm:no-EF-PO-with-n-greater-than-2}, we know that this result is impossible for a fixed number of rounds, and thus highlights the difference between fixed and variable~$k$.

%%%%%%%%%%%%%%%%%%%%%%%%%%%%%%%%%%%%%%%%%%%%%%%%%%%%%%%%%%%%%%%%%%%%%%%%%%%%%%%%%%%%%%%%%%%%%%%%%
\section{Conclusion}\label{sec:conclusion}
%%%%%%%%%%%%%%%%%%%%%%%%%%%%%%%%%%%%%%%%%%%%%%%%%%%%%%%%%%%%%%%%%%%%%%%%%%%%%%%%%%%%%%%%%%%%%%%%%

We have seen that in our model of repeated allocations the
(necessary) trade-off between fairness and efficiency is much more favorable 
than in the standard setting without repetitions.
In the case of two agents, we presented some algorithms
guaranteeing overall envy-freeness and
Pareto-optimality (as well as per-round approximate
envy-freeness) for any even number of rounds.

As some of our algorithms require exponential time, it would be of interest
to study the computational complexity of related decision problems, investigating 
whether and where polynomial-time results are obtainable.

Our $n$-agent algorithm yields slightly weaker guarantees (proportionality and Pareto-optimality),
which are still an improvement over the one-shot setting.
It remains for future work to determine whether this result can be strengthened by
additional per-round guarantees, as for the 2-agent case.

When the number of rounds $k$ can be chosen freely, we have shown that (for any number of agents) an envy-free, Pareto-optimal
and per-round PROP[1,1] allocation always exists. We have done so by establishing a connection between our setting
and the probabilistic and divisible settings. However, our approach gives no guarantee on the number of rounds it requires.
As future work, it would be interesting to investigate the complexity of finding the smallest $k$
for which an envy-free and Pareto-optimal allocation exists.

%%%%%%%%%%%%%%%%%%%%%%%%%%%%%%%%%%%%%%%%%%%%%%%%%%%%%%%%%%%%%%%%%%%%%%%%%%%%%%%%%%%%%%%%%%%%%%%%%

\section*{Acknowledgments}
This paper developed from ideas discussed at the Dagstuhl Seminar 22271 \emph{Algorithms for Participatory Democracy}. Martin Lackner was supported by the Austrian Science Fund (FWF): P31890. Oliviero Nardi was supported by the European Union's Horizon 2020 research and innovation programme under grant agreement number \includegraphics[height=\fontcharht\font`\B]{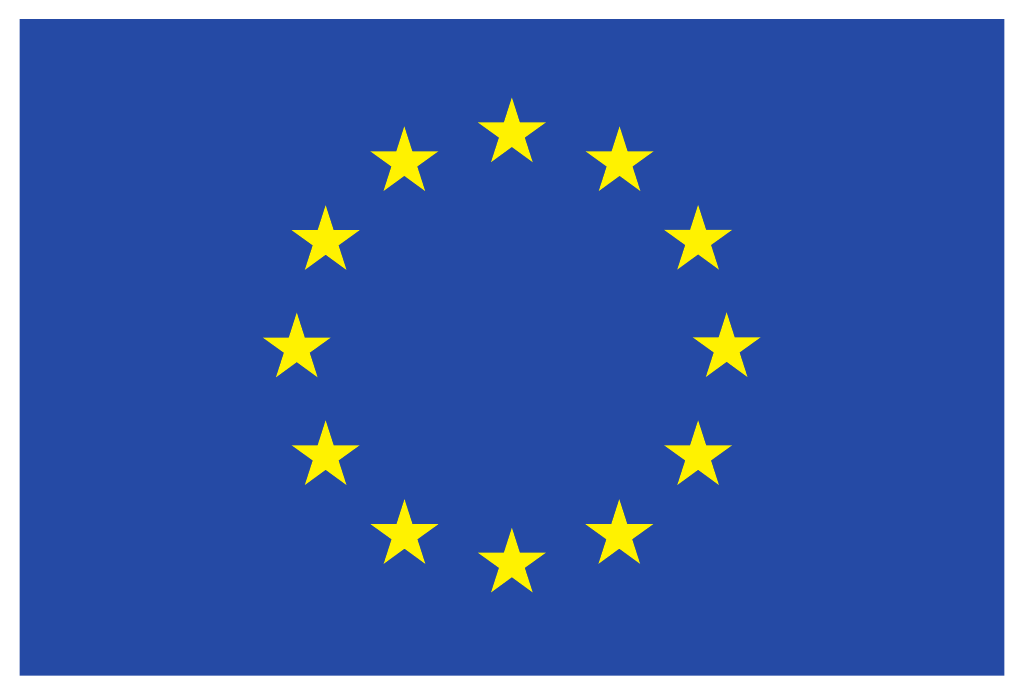}101034440, and by the Vienna Science and Technology Fund (WWTF) through project ICT19-065. Ayumi Igarashi was supported by JST PRESTO under grant number JPMJPR20C1.

\balance
\bibliographystyle{ACM-Reference-Format}

\end{document}